\documentclass[runningheads]{llncs}

\usepackage[T1]{fontenc}
\usepackage{graphicx}

\usepackage{hyperref}
\usepackage{color}

\usepackage{amsfonts,amsmath,amsthm}
\usepackage{url}
\usepackage{lineno}

\usepackage{todonotes}
\usepackage{algorithm,algpseudocode,booktabs}
\usepackage{tikz}
\usepackage{pgfplots,pgfplotstable}
\usepackage{csvsimple}
\usepackage{bbm}
\usepackage{esvect}
\usepackage{multirow}
\usepackage{amssymb}
\usepackage[normalem]{ulem}
 \usepackage{subcaption}

\usepackage{xspace} 
\usepackage{xparse} 

\newcommand\newmath[2]{\newcommand#1{\ensuremath{#2}\xspace}}
\newcommand\renewmath[2]{\renewcommand#1{\ensuremath{#2}\xspace}}

\newcommand\newmathope[2]{\newcommand#1{\ensuremath{\operatornamewithlimits{#2}}\xspace}}

\usepackage{tikz,pgf}
\usetikzlibrary{arrows,automata,shapes,calc,positioning,intersections,backgrounds}

\newmath{\N}{\mathbb{N}}
\newmath{\Z}{\mathbb{Z}}
\newmath{\Q}{\mathbb{Q}}
\newmath{\R}{\mathbb{R}}

\newmathope{\argmin}{\arg\min}
\newmathope{\argmax}{\arg\max}

\renewmath{\Pr}{\mathbb P}
\newmath{\PP}{\mathbf P}
\newmath{\Bad}{\mathsf{Bad}}
\newmath{\Opt}{\mathsf{Opt}}
\newmath{\Dist}{\mathcal D}
\newmath{\M}{\mathcal M}
\newmath{\Supp}{\mathsf{Supp}}
\newmath{\E}{\mathbb E}
\newmath{\UPre}{\mathsf{UPre}}
\newmath{\Good}{\mathsf{Good}}
\newmath{\Ugly}{\mathsf{Ugly}}

\newmath{\Rew}{\mathsf{Loss}}
\newmath{\AReward}{\mathsf{ARew}}

\newmath{\FPaths}{\mathsf{Paths}}
\newmath{\GoodPaths}{\mathsf{GoodPaths}}
\newmath{\BadPaths}{\mathsf{BadPaths}}
\newmath{\first}{\mathsf{first}}
\newmath{\second}{\mathsf{second}}
\newmath{\last}{\mathsf{last}}
\newmath{\len}{\operatorname{len}}
\newmath{\States}{\mathsf{States}}
\newmath{\Val}{\mathsf{Val}}
\newmath{\Win}{\mathsf{Win}}
\newmath{\ValSafety}{\mathsf{ValSafety}}
\newmath{\ValReach}{\mathsf{ValReach}}
\newmath{\Shield}{\mathsf{Shield}}

\newmath{\opt}{\mathsf{opt}}
\newmath{\Cyl}{\mathsf{Cyl}}
\newmath{\GoodCyl}{\mathsf{GoodCyl}}

\newmath{\NN}{\mathsf{NN}}

\newmath{\reward}{\mathsf{reward}}
\newmath{\numsamples}{\mathsf{count}} 
\newmath{\children}{\mathsf{children}}
\newmath{\mctsvalue}{\mathsf{value}}
\newmath{\total}{\mathsf{total}}
\newmath{\I}{\mathcal{I}}
\newmath{\iter}{\mathsf{iter}}

\newmath{\A}{\mathscr A}
\newmath{\Apsi}{\A^{\psi}}
\newmath{\Aphi}{\A^{\varphi}}

\newmath{\cyl}{\mathsf{Cyl}}
\newmath{\invalpha}{\alpha^{\text{-}1}}

\newcommand{\cont}{\mathsf{\neg STOP}}

\newcommand{\maxi}{\mathsf{maxi}}
\renewcommand{\stop}{\mathsf{STOP}}

\allowdisplaybreaks

\usepackage{thm-restate,thmtools}

\begin{document}

\title{Algorithms for Robbins' Problem using \\ Markov Decision Processes\thanks{This work was primarily conducted when Dr. Anirban Majumdar was a post-doctoral researcher at ULB, and Dr. Léonard Brice was a doctoral student at ULB.}
}
\titlerunning{Algorithms for Robbins’ Problem using MDPs}
%
\author{Léonard Brice\inst{1}
\and F. Thomas Bruss\inst{2} \and \\
Anirban Majumdar\inst{3}
\and  Jean-Fran\c{c}ois Raskin\inst{2}
}

\institute{Institute of Science and Technology Austria, Klosterneuburg, Austria \email{leonard.brice@ista.ac.at}
\and
Université Libre de Bruxelles, Bruxelles, Belgium
\email{\{thomas.bruss,jean-francois.raskin\}@ulb.be}
\and
Tata Institute of Fundamental Research, Mumbai, India
\email{anirban.majumdar@tifr.res.in}
}
\authorrunning{Brice, Bruss, Majumdar and Raskin}
\maketitle              
\begin{abstract}
In this paper, we consider Robbins' problem, which is a full information variant of the well-known {\em secretary selection problem}. In this version of the problem, the goal is to minimize the expected rank of the selected candidate among $n$ that are interviewed sequentially, and a decision to select or not the $m^{th}$ candidate needs to be taken right after the interview (so without seeing the last $n-m$ candidates and without recall). We first show how to model instances of Robbins' problem as infinite Markov Decision Processes (MDPs). Then we propose several finite-state abstractions of these MDPs that allow us to approximate the value of the problem for fixed $n$. While it is known that the full memory of past candidates' values is necessary for optimal expected rank minimization, making the analysis of the problem challenging, we highlight simple memory structures that are sufficient for obtaining near-optimal selection strategies. Additionally, we provide approximate values for Robbins' problem for numbers of candidates $n$ up to 100 for which no good approximations were previously known (the exact value is only known for instances where $n \leq 4$ and numerical approximations were for small values of $n$ not exceeding one digit), for all $n : 5 \leq n \leq 100$, we give better approximation than what was previously known.
\end{abstract}



\section{Introduction}
\label{sec:intro}

Imagine the high-stakes world of professional cycling, where a director of a top-tier team is on the hunt for the perfect new team member. There are a known number $n$ of candidates, and the director must interview them one by one. During the interview, the director asks questions about the candidate's weight, ${\sf FTP}$ number (Joost-Pieter Katoen, to whom this paper is dedicated, will know what this acronym means), ${\sf VO}_2{\sf max}$ value, and more. The director is exceptionally precise and has a function that maps the candidate's skills and attributes to the interval $[0,1)$, and this function returns small values for the best cyclists, and it allows to compare the cyclists in a fine-grained way\footnote{Candidates like M(v)DP, pun intended, would be mapped very close to $0$ (best value) and should hopefully be selected by the director.}. For example, this allows him to \emph{rank} the current candidate against all previous ones, and as a specialist (the best candidate having rank $1$, and the worst one having rank $n$), he also knows the probability distribution of these values among the set of candidates being interviewed as they are assumed to be randomly taken from a set of cyclists with world tour level.

Because other teams are also looking for the best cyclists, after each interview, the director must make an immediate decision about whether to hire the cyclist and add this candidate to his team, with no opportunity for second-guessing. Furthermore, if the first $n-1$ candidates are rejected, the director is obliged to hire the last one.

The challenge is to devise a \emph{selection strategy} that minimizes the expected rank of the selected cyclist among the $n$ cyclists that he may potentially interview. In essence, how can the director ensure he is getting the best possible additional team member based on the candidate's rank in expectation among the $n$ cyclists that he can review? And for the more mathematically oriented director (like Joost-Pieter Katoen would be), he might intriguingly ask: what is the limiting value of this expected rank as the number of candidates becomes large, \emph{i.e.}, tends to infinite?

The problem that we sketched above not only tests the director's decision-making skills, but also his ability to strategize under pressure, ensuring the team gets a top-notch cyclist who can help them race to victory. This is a variant of the \emph{secretary selection problem} (the fourth variant~\cite{bruss2005}), and it is also known as \emph{Robbins' problem}. It still remains open to find an optimal strategy for this problem, as well as to determine the limiting value when $n$ tends to infinity.

 Robbins  presented this problem at the end of his memorable talk on the International Conference on Optimal Stopping and Selection in 1990 (U. of Massachusetts). It  is deep and an easy-to-describe representative of a whole class of problems in Probability Theory, which attract interest for several reasons. This is the class of problems of full history-dependence.
Before formally defining Robbins' problem, let us intuitively explain the notion of a strategy on the previous example of selecting the best cyclist.

Here we could assume that the director will apply the following strategy:
when interviewing the $i$-th candidate (with $1 \le i \le n$), if the \emph{measure} of the candidate is less than or equal to $2/{n-i+2}$, then select that candidate. The above strategy is an instance of a more general class of strategies $\sigma_n^c$: select the first candidate whose measure is less than or equal to $c/{n-i+c}$.
Let us elaborate this strategy on a particular instance of the cyclist-selection problem described above. 

Assume that the director has 4 interviews scheduled.
At the first interview, assume that the value of the candidate is $0.7$. Recall that small values in $[0,1)$ are associated to the best cyclists, this one is rather average and as
$$0.7 > 2/{4-1+2} = 2/5,$$
the director bets to see a better candidate in the remaining $3$ interviews, and therefore decides to "not stop".

Now assume that the measure of the second candidate is $0.55$. Again, since 
$$0.55 > 2/{4-2+2} = 2/4,$$
the director decides to "not stop" and wait for hopefully better candidates to come.

Finally, assume that the third candidate has measure $1/3$. Since
$$1/3 < 2/{4-3+2} = 2/3,$$
the team director decides to "stop" and select this candidate.

Let us note a few remarks. First, let us note that, by stopping now, the director chooses the best candidate seen so far (in this example), so the candidate's \emph{relative rank} is $1$. However, the team director has not seen the last candidate, and hence, the right question to ask is the following: what is the expected rank of the selected ($3$rd) candidate knowing that the last candidate is taken uniformly at random from $[0,1)$. According to that, the probability that the last candidate is better than the selected one is $1/3$ (in that case, the \emph{absolute rank}, or simply, \emph{rank} of the selected candidate is $2$), and the probability that it is worse than the selected one is $2/3$ (in that case, the rank of the selected candidate is $1$). Therefore, the expected rank of the selected candidate is:
$$1/3.2 + 2/3.1 = 4/3 = 1.33.$$

Now, let us take the time here to note that it is also mathematically relevant to ask the following question: what is the expected rank of the candidate selected by the strategy $\sigma_n^{c=2}$. Further, one can also ask how this expected rank evolve  as $n \to \infty$. Both of these problems for the  class of strategies $\sigma_n^{c}$ have been studied in details in the paper~\cite{bruss1993minimizing}. 
We also note that the straightforward strategy described above does not utilize the entire history of the values of the candidates observed so far. Consequently, this strategy cannot be optimal regardless of the value of $c$, since, as shown in~\cite{bruss1996}, an optimal strategy for the Robbins' problem needs to fully depend on the history of the measure of the candidates seen so far.

We now formally define the Robbins' problem and the relevant notations.

\paragraph{{\bf Definition of Robbins' problem.}}
Let $n\in \N$ be a fixed natural number, and let $X_1, X_2, \cdots, X_n$ be independent and identically distributed (i.i.d.) random variables uniform on $[0,1).$
We can observe them sequentially, \emph{i.e.}, in the order $X_1, X_2, \cdots,$ and we must select {\it exactly one} of them. A selection is only possible at the time of observation, and if $X_m$ is selected, then the  decision is irrevocable
and the game is finished. At step $n$ we see the whole picture, and if we have selected $X_m$, with $1\le m \le n$, then we occur the  loss
\begin{align} 
L_m=\sum_{j=1}^n \mathbbm{1}\{X_j \le X_m\}, 
\end{align} 
where $\mathbbm{1}\{A\}$ is the \emph{indicator function} of the event $A$, formally: $\mathbbm{1}\{x\} = 1$ if $x \in A$, and otherwise $\mathbbm{1}\{x\} = 0$.

In other words, if we denote the increasing order statistics of the $X_j$ by
\begin{align}
X_{1,n} \le X_{2,n} \le  \cdots \le X_{n,n}
\end{align} 
and if we have chosen $X=X_{j,n},$ then our loss is $j$.

Note that we see the {\it values} and not only their relative ranks, but the {\it loss is the rank} of the accepted observation!  At time $n$, the random variable $L_m$ becomes deterministic, taking as value the final rank of $X_m$
among the whole sample $X_1, X_2, \cdots, X_n$. If we must choose exactly one variable,
what sequential strategy will minimize the expected loss?\footnote{Robbins announced this problem by saying ``Finally, here is the problem which I'd like to see solved before I die''. Robbins'  wish did not realize. He died February 12th, 2001, and unfortunately the main part of the problem is still open today.}

A \emph{strategy} $\sigma_n$, for a fixed $n \in \N$, is a function that, given a sequence of numbers from $[0,1)$ of length less than or equal to $n$, decides to accept the last value (or, to continue) using $\stop$ (resp.,~$\cont$), with the additional condition that $\sigma_n$ must accept exactly one number along all possible $n$-length sequences. Formally, 
$$\sigma_n : [0,1)^{\le n}\to \{\stop, \cont\},$$ 
such that $\forall x_1, \ldots, x_n \in [0,1)^n$: 
$\exists j \le n: \sigma(x_1 \ldots x_j) = \stop$; 
and for all $j' < j, \sigma(x_1 \ldots x_{j'}) = \cont$.
A \emph{solution} of Robbins' problem is then an \emph{optimal strategy} $\sigma_n^*$ that minimizes the expected loss: 
$$\E^{\sigma_n^*}(L) = \inf_{\sigma_n} \E^{\sigma_n}(L) = v_n.$$

We also note that solving Robbins' problem for the uniform distribution will also mean solving it for any absolute continuous distribution $F$. Indeed, then we have a one-to-one correspondence between $X_1,\cdots,X_n$ and $F(X_1), \cdots, F(X_n)$.
Since $\Pr(F(X) \le y) = \Pr(X \le F^{-1}(y))$, and the latter follows the distribution function $F$, it equals $F(F^{-1}(y)) = y$, \emph{i.e.}, $F(X)$ is uniform on $[0,1)$.

\paragraph{{\bf The challenge of Robbins' problem as a secretary problem.}}
\label{par:secretary-problems}
We briefly explain why Robbins' problem stands out in the class of so-called secretary problems.
The overall probably best-known optimal stopping problem is the so-called (classical)
secretary problem, or "first" secretary problem. This is to maximize the probability of selecting with a single choice, with  no recall of preceding observations and with rank information, the
"best" rank (rank one) from a sequence of $n$ uniquely rankable objects. Its solution was published in~\cite{lindley1961}. 
The same problem with full information (\emph{i.e.}, the values and not only their ranks could be observed) was solved in~\cite{gilbert1966}. Minimizing the expected rank under rank information was proved to be harder, and was completely solved in~\cite{chow1964}. The same problem under full information is the fourth secretary problem, rounding up this two-by-two design of optimization problems, and this is Robbins' problem.  We know that its full solution is open, so the  "fourth" problem already presents a particular challenge in the class of secretary problems.

\paragraph{{\bf Known results for Robbins' problem.}}\label{BAna-RP}
Since for each $n$,  we need only to take $n$ decisions, by backward induction principle, for each $n$, an optimal strategy must exist. Let $v_n$ denote the corresponding optimal value, \emph{i.e.}, the minimal expected (final) rank.

\smallskip
The following results ((i) to (iv)) are proved in \cite{bruss1993minimizing,chow1964,assaf1996secretary}. 

\begin{itemize}
    \item [(i)]
$v_n$ is increasing in $n.$

\item[(ii)]
The sequence $(v_n)$ is bounded below (trivially by $L=1$) and bounded above by another known value from a related problem studied in \cite{chow1964}, which is $3.869\cdots$. This is one of the variants of the secretary problem (as described in the previous paragraph), 
where only the rank of the new draw within the current history is revealed and not the exact value. Since seeing the numerical values of the random variables allows us to rank them, obtaining the numerical values provides at least as much information as rank information. Therefore, one cannot achieve better results with rank information than by seeing the $X_j$. The optimal expected rank for rank information is $\tilde v=3.869\cdots$, which is the upper bound mentioned above. As a direct consequence, the following limit is known to exist:
\begin{align}
\tilde v = \lim_{n \to \infty} \tilde v_n \quad \text{exists and} \quad v = \lim_{n \to \infty} v_n \leq \tilde v,
\end{align}
where $\tilde{v}_n$ denotes the optimal value for $n$ draws under rank-information.
Note that the same reasoning applies to obtain an upper bound using memoryless strategies. This approach is summarized in the next paragraph. 

\item[(iii)]
Smaller $X_m$ have smaller ranks. Hence, it is intuitive that the values $X_m$ and the corresponding final ranks $L_m$ (which will be known at time $n$ only) should be positively correlated. Recall that {\it correlation} is a measure of dependence of one random variable of another one. We can compute it  and obtain (\cite{bruss1993minimizing}, (1.6) - (1.8)):
\begin{align}\label{corr}{\rm as ~}n\to\infty,~ \forall \,1\le m\le n:  {\rm corr}(X_m, L_m)= \sqrt{\frac{n-1}{n+1}} \to 1 .\end{align}
The strong positive correlation described in Equation~\ref{corr} suggests proposing a strategy by just looking at time $m$ at the observed value $X_m$ and to select it if and only if $X_m$ is smaller or equal to some threshold $\varphi(X_1, \cdots, X_{m-1}; X_m; n)$. This is a {\em threshold strategy}. If, moreover, we ignore at each step $m$ all preceding values $X_1, \cdots, X_{m-1},$ then we speak of a {\it memoryless} threshold strategy\footnote{In the community of applied probability, where  Robbins' problem originates, such strategies are referred to as {\em memoryless threshold strategies}. However, in the Formal Methods community, we refer to them as {\em counting strategies}, as the threshold used depends on the round number.}, or in short  {\it ml}-strategy. We denote the optimal value obtainable for $n$ observations in this restrained class by $\overline{v}_n.$
This strong positive correlation of the values $X_1, X_2, \cdots, X_n$ and their corresponding absolute ranks can be exploited to define a pattern of memoryless threshold strategies.  The precise optimal memoryless strategy is complicated to compute (see \cite{assaf1996secretary}); however, it is known that a threshold strategy using the threshold function $\varphi(n,m) = \frac{c}{n - m + c}$, \emph{i.e.}, the $m^{th}$ is accepted if this draw is less than or equal to $\frac{c}{n - m + c}$, with $c = 1.9469\cdots$ yields a good approximation. The upper bound obtained with this memoryless strategy is equal to $2.3318\cdots$ when $n$ tends to infinity. So this improves the upper bound that we obtain with $\tilde v$.

It is however important to note that it was shown that the memoryless optimal strategy can always be improved, see e.g.,~\cite{meier2017}. However, so far nobody has proved that the improvement will not disappear as $n\to \infty$. The author in~\cite{gnedin2007} studied a related problem, namely \emph{Poisson embedded Robbins' problem}, in which the optimal limiting strategy is indeed not a memoryless threshold strategy.

\item[(iv)]
Lower bounds of different levels for $v_n,$ and thus for $v$ can be obtained by a truncation argument (see~\cite[Section~$4$]{bruss1993minimizing}). We say we truncate the loss at level $j$ with $1 \le j\le n,~ j\in \N,$ if the loss generated by $X_m$ is defined to be equal to $\min\{j, L_m\}$. Therefore, clearly,
\begin{align}\label{BAna-RP(vi)}\forall 1\le j\le n: v_n(j) \le v_n.\end{align}
Truncating at the level $j=5,$ the authors in \cite{bruss1993minimizing} obtained the lower bound $L\approx 1.908$. 
\end{itemize}

\paragraph{{\bf  Intrinsic difficulties of Robbins' problem.}}

We seemingly cannot compute $v_n$ for large $n$ as we 
 do not know the optimal strategy in general. It is trivial for $n=1$, almost trivial for $n=2,$ and still easy for $n=3$, but that's it as far as we can say "easy".  The problem is that the optimal strategy is {\it fully history dependent}, which means that the optimal strategy depends at each step $m$ on the complete preceding cloud of values  $X_1, X_2, \cdots, X_{m-1}$. The  order in which the points arrived is irrelevant.
This is why it is the "cloud" $\{X_1, X_2, \cdots, X_{m-1}\}$ of the history which counts. The full history dependence is proved in~\cite{bruss1996}.

  \smallskip
From a decision-theoretical point of view, the interpretation of full history-dependence is as follows. There is no sufficient statistic for optimal decisions other than the whole history itself, that is, the cloud of all points in $[0,1)$ seen so far. This is not necessarily a serious problem for classes of problems where the influence of the history on the currently observed process can be shown to become irrelevant sufficiently quickly.
However, Robbins' problem is a representative in the class of fully history-dependent problems where this asymptotic irrelevance is not clear.

\smallskip
To understand the difficulty to compute the optimal strategy for $n$ points and the corresponding value $v_n$ precisely, we refer to Figure 1-Figure 4 of \cite{n=4}. The authors of \cite{n=4} coped with the challenge to solve the problem precisely for $n=4$. Their graphs for the composed acceptance regions of the corresponding optimal strategy are complicated and do not show an easy structure.

 \medskip
The idea
to compute, for larger $n$, the $v_n$  by the truncation method, described in (iv) of Section \ref{BAna-RP}, is seemingly hopeless. As shown in \cite{bruss1993minimizing}, the storage demands to implement the truncation method increases exponentially in both the number $n$ and the truncation level $j.$ Under the hypothesis that the capacity of computers increases exponentially, this is far away from the double-exponential growth we would need.

\smallskip
These combined difficulties may explain why many authors have stopped working on Robbins' problem.
Indeed, not much was contributed to Robbins' problem after the nineties.

In a recent paper~\cite{bruss2024}, several aspects of this full history dependence,  as well as their implications,  are investigated in detail. The author concludes that the intricacies of the problem require a new approach and argues that deep learning is, for several reasons, a promising candidate for this goal. However, the paper does not contain a  concrete description of the architecture of a suitable neural network. This contrasts our present paper essentially, because our approach, based on Markov decision processes, is \emph{concrete} throughout, and enables us to obtain new values of Robbins' problem for a range of interesting $n$.

\paragraph{{\bf Contributions.}} 
In this paper, we consider several abstractions of  Robbins' problem that lead to finite MDPs solvable with algorithms based on backward induction, as implemented in tools like \textsc{Storm}~\cite{storm2017}. 
Note that, if we consider the clouds simply as the successive states of the decision process, then this sequence indeed forms an MDP.
These finite MDPs are useful for at least two purposes. First, when $n$ is fixed, they allow us to compute an upper bound of the value $v_n$. We will demonstrate in the experimental section that the upper bounds obtained this way are better than the currently known upper bounds, which are derived from memoryless strategies. Second, by solving the underlying MDPs, we obtain strategies that can be practically implemented and whose performances surpass those of the previously known memoryless strategies. Furthermore, from these abstractions, we will highlight some properties that are crucial in the history for playing a good strategy. 

\smallskip
\noindent\textit{Discretization.}
The first abstraction is a straightforward discretization of the problem: we partition the interval $[0,1)$ in $d \in \mathbb{N}$ intervals of length $\frac{1}{d}$ and we abstract the draws by the interval in which they fall, and histories  by counting how many draws landed in each interval. We show that for a fixed $n$, when $d$ tends to infinity, this approximation scheme is giving the value $v_n$. The number of states of the underlying MDP increase proportionally in the number of draws and the number of possible histories of draws, the latter is given by the binomial coefficient $\binom{n+d-1}{n}$. As a consequence, this abstraction can be used only to approach the optimal value $v_n$ for small values of $n$ and $d$. Still, this MDP can be used to closely approximate $v_n$ for values of $n$ for which no good approximations were known so far. We report on this in the experiments section of the paper.

Note that, $\binom{n+d-1}{n}$ corresponds to the number of possibilities to place, without constraints, $n$ balls in $d$ places. If we limit our interest to single occupations (\emph{i.e.}, conflict-free) only, the number becomes $\binom{d}{n}$. Then it is easy to see that $\binom{d}{n}/\binom{n+d-1}{n} \to 1$ as $d$ goes to infinity. In other words, as $d$ grows, the non-conflict probability tends to 1. This argument can be used to show that this simple abstraction scheme preserves the limit value (see Theorem~\ref{thm:limit-completness}).

\smallskip
\noindent\textit{Limiting the history.}
In the second abstraction, we limit the histories, and so the memory, in the following way: instead of remembering the number of draws that fall into every interval in the history, we only consider the intervals of the $k$ best previous draws. This results in an MDP whose size is bounded by ${\cal O}(n \cdot d^{k+1})$. While this MDP only achieves the precision of the first approach when $k=n$, we observe experimentally that the value obtained using this MDP, which is proven to over-approximate the true value of  Robbins' problem, is close to the value computed with the first MDP even for small values of $k$ (e.g., $k=2$ or $k=3$). This implies that remembering only the values of the $k$ best previous candidates allows us to devise a good strategy for stopping on a candidate with a small expected rank.

\smallskip
\noindent\textit{Handling the bad draws carefully.}
In the third abstraction, we limit the way we record bad draws. In the finite-state MDP that keeps track of the $k$-best draws seen so far, defined in Section~\ref{sec:k-abs}, the number of states grows as ${\cal O}(n \cdot d^{k+1})$, making it difficult to consider large values for $d$.
As $n$ grows, $d$ must also be reasonably large to  approximate $v_n$ sufficiently well. Otherwise, the probability of conflicts in histories increases, reducing precision. Therefore, computing the value for the underlying MDP may become infeasible. However, if we remember the $k$ best draws, it is likely that the values in memory are "good candidates" (\emph{i.e.}, in $[0,l)$). For instance, if we remember the $k=2$ best draws and have seen $j$ draws taken uniformly at random, the probability that one is above $l$ becomes very small when $j$ is growing. Thus, for large $n$, it is reasonable to remember only the exact interval of the "good candidates" seen so far. Our next abstraction uses a non-uniform partitioning of $[0,1)$, with a precise partitioning for $[0,l)$ into intervals of size $\frac{1}{d}$, and only one interval for candidates in $[l,1)$. This reduces the number of states of the MDP to ${\cal O}(n \cdot l^{k+1})$. In Section~\ref{sec:experiments}, we will report on values for several choices of $l$, showing that, in practice, as $n$ grows, a relatively smaller $l$ suffices for a reasonable approximation of the value $v_n$. 

\paragraph{{\bf Structure of the paper.}} In Section~\ref{sec:prelims}, we introduce the necessary preliminaries. In Section~\ref{sec:d-abs}, we present the MDP of the first abstraction, which discretizes the interval $[0,1)$ into intervals of length $\frac{1}{d}$. In Section~\ref{sec:k-abs}, we introduce the second abstraction and its MDP formalization, in which only the interval of the $k$ best candidates are memorized in the history. In Section~\ref{sec:(k,l)-abs}, comes the third abstraction, in which only the interval of {\em good} candidates (those whose value is less than $l$) are precisely recorded. In Section~\ref{sec:experiments}, we present experimental results. In Section~\ref{sec:conclusion}, we draw  conclusions and hint on possible future works.

Missing proofs can be found in Appendix~\ref{sec:appendix}. All values for $n = 1$ to $100$ with $d = 500$, $1000$ and $k = 2$, and for the memoryless strategy of~\cite{assaf1996secretary} are also reported in Appendix~\ref{sec:appendix}.

\section{Preliminaries}
\label{sec:prelims}

We introduce here the definitions for Markov chains and Markov decision processes. The interested reader will find more about those models in~\cite{DBLP:books/wi/Puterman94,DBLP:books/daglib/0020348}.

A \emph{probability distribution} on a finite set $S$ is a function $\delta:S\to [0,1]$ such that $\sum_{s\in S}\delta(s)=1$.
We denote the set of all probability distributions on set $S$ by $\Dist(S)$. The support of a distribution $\delta\in \Dist(S)$ is $\Supp(\delta)=\{s\in S\mid \delta(s)>0\}$.

\begin{definition}[Markov chain]\label{def:mc}
	A (discrete-time)  Markov chain or an MC is a tuple $M=(S,s_{in},P,\Rew)$, where
	 $S$ is a  set of states,
  $s_{in} \in S$ is the initial state, 
		$P$ is a mapping from $S$ to $\Dist(S)$,
        and $\Rew$ is a partial mapping $\Rew: S\to \R$. 
\end{definition}

For states $s,s'\in S$, $P(s)(s')$ denotes the probability of moving from state $s$ to state $s'$ in a single transition, and we denote this probability $P(s)(s')$ as $P(s,s')$. A state $s \in S$ for which the loss mapping $\Rew$ is defined is called {\em final}. We assume here that for all final states $s$, $P(s,s)=1$ (\emph{i.e.}, final states are \emph{absorbing}). We denote by $\FPaths^{\Rew}_M$ the following set of finite paths $s_0 s_1 \dots s_n$ in $M$ such that:
  \begin{itemize}   
    \item $s_0=s_{in}$,
    \item $P(s_i,s_{i+1})>0$, for all $i:0 \leq i <n$,
    \item $s_n$ is a final state, \emph{i.e.}, $\Rew(s_n)$ is defined.
  \end{itemize}
So $\FPaths^{\Rew}_M$ is the set of  paths in $M$ that start in the initial state of $M$, and reach a final state. The probability of a path in $\rho=s_0 s_1 \dots s_n \in \FPaths^{\Rew}_M$ is equal to ${\sf Prob}(\rho)=\prod_{i=0}^{n-1} P(s_i,s_{i+1})$. Given $M$, we are interested in the expected loss obtained when executing $M$, which is noted $\mathbb{E}^M(\Rew)$ and equal to: 
$$\sum_{\rho \in \FPaths^{\Rew}_M} {\sf Prob}(\rho) \cdot \Rew({\sf last}(\rho))$$
\noindent
where ${\sf last}(\rho)$ denotes the final state of $\rho$.

\begin{definition}[Markov decision process]\label{def:mdp}
	A Markov decision process or an MDP is a tuple $\M=(S,A,s_{in},P,\Rew)$, where
	 $S$ is a  set of states,
	 $A$ is a finite set of actions,
        $s_{in} \in S$ is the initial state,
	 $P$ is a (partial) mapping from $S\times A$ to $\Dist(S)$, and \Rew is a partial mapping $\Rew: S\times A\to \R$.
\end{definition}
$P(s,a)(s')$ denotes the probability that action $a$ in state $s$ leads to state $s'$ and we denote this probability $P(s,a)(s')$ as $P(s,a,s')$. Therefore, if an action $a$ is admissible from a state $s$, we will have $\sum_{s'\in S} P(s,a,s')=1$. Otherwise, we will have $P(s,a,s')$ is undefined (denoted by $\bot$) for all $s'\in S$. 
A \emph{finite path} $\rho = s_0 a_0 s_1\ldots a_{i-1} s_i$ is a sequence of states and actions such that for all $t\in[0,i-1]$, we have $s_{t+1}\in \Supp(P(s_t, a_t))$. We denoted by $\FPaths_{\M}$ the set of all finite paths in $\M$.

For an MDP $\M$, a \emph{strategy} is a function $\sigma : \FPaths_{\M} \to A$ that maps a finite path $\rho$ to an action $a \in A$. 
Note that a strategy ${\sigma}$ in an MDP induces an MC $\M_{\sigma}$. Intuitively, this MC is obtained by unfolding $\M$ using the strategy $\sigma$ and using the probabilities in $\M$ to define the transition probabilities. Formally, $\M_{\sigma} = (\FPaths_{\M},s_{in},P_{\sigma}, \Rew)$ where for all  $\rho\in \FPaths_{\M}$, $P_{\sigma}(\rho)(\rho\cdot as) =  P(\last(\rho),a)(s)$,  if $\sigma(\rho)=a$, and equals to $0$ otherwise.
MDPs are sometimes coined as $1\frac{1}{2}$ player games in the literature. Accordingly, we call the agent that take the decisions, and so execute a strategy $\sigma$ along the execution the {\em protagonist}.

For a  history $h$, we write $\sigma_{\|h}$ for the strategy $h' \mapsto \sigma(hh')$.
We also, abusing notation, write $\E(\sigma)$ for the expected loss $\E^{\M_\sigma}(\Rew)$ obtained by the protagonist when using the strategy $\sigma$, that is, the expectation of loss in the MC $\M_{\sigma}$.

\section{An MDP abstraction for full $d$-discrete history}
\label{sec:d-abs}
\paragraph{{\bf The exact Robbins' problem MDP.}}
 Robbins' problem can be formalized as an (uncountable infinite) MDP, denoted $\M_n$, where states are the histories of draws received so far, and the set of actions is ${ {\sf STOP}, \neg {\sf STOP} }$. Whenever the $\neg {\sf STOP}$ action is played, a new point is drawn uniformly at random from the interval $[0,1)$ and added to the history of draws. When the {\sf STOP} action is played, or when the last draw has been drawn, the expected rank associated with the last draw is returned as a loss. Given a history $h$ and the draw $x_i$ where the action ${\sf STOP}$ is chosen, we can compute the expected loss by determining the current rank of $x_i$ within the history and adding the expected number of remaining draws that will be better than $x_i$. Solving this MDP involves finding the best strategy to minimize this expected rank.
A part of this MDP is shown in Figure~\ref{fig:cont-MDP}.

Unfortunately, this MDP contains an uncountable number of possible states, making it impractical to solve directly. To obtain an MDP with a finite number of states that can be solved algorithmically, we partition the interval $[0,1)$ into $d \in \mathbb{N}$ intervals of length $\frac{1}{d}$ and use the following function to discretize histories:

\begin{definition}
    Let $d,n \in \mathbb{N}$, the $d$-discretization function 
    $\alpha_d: [0,1)^{\leq n} \to \mathbb{N}^d \cap [0,n]^d$  maps histories with $l \leq n$ draws to vectors of $d$ natural numbers as follows: let $h=x_1, x_2, \cdots, x_l$ in $[0,1)^l$, 
        $\alpha_d(h) := \vv{y},$
     such that for every $0 \le i \le d-1$:
     $$\vv{y}[i] = \left |\left \{j \mid 1 \le j \le l \text{ and } x_j \in [\frac{i}{d}, \frac{i+1}{d} ) \right \} \right|.$$
\end{definition}

In the finite discrete abstraction, each state of the MDP will be a vector $\vv{y} \in \mathbb{N}^d \cap [0,n]^d$ that abstracts the histories so far by counting, for each interval $I$ of length $\frac{1}{d}$, the number of draws in $I$, along with the interval of the last draw received and the number of draws remaining. As with the infinite MDP, the set of actions is ${ {\sf STOP}, \neg {\sf STOP} }$. We formally define the MDP as follows:

\begin{figure}[t]
  \centering
  \begin{subfigure}{.45\textwidth}
  \centering
      \includegraphics[scale=0.25]{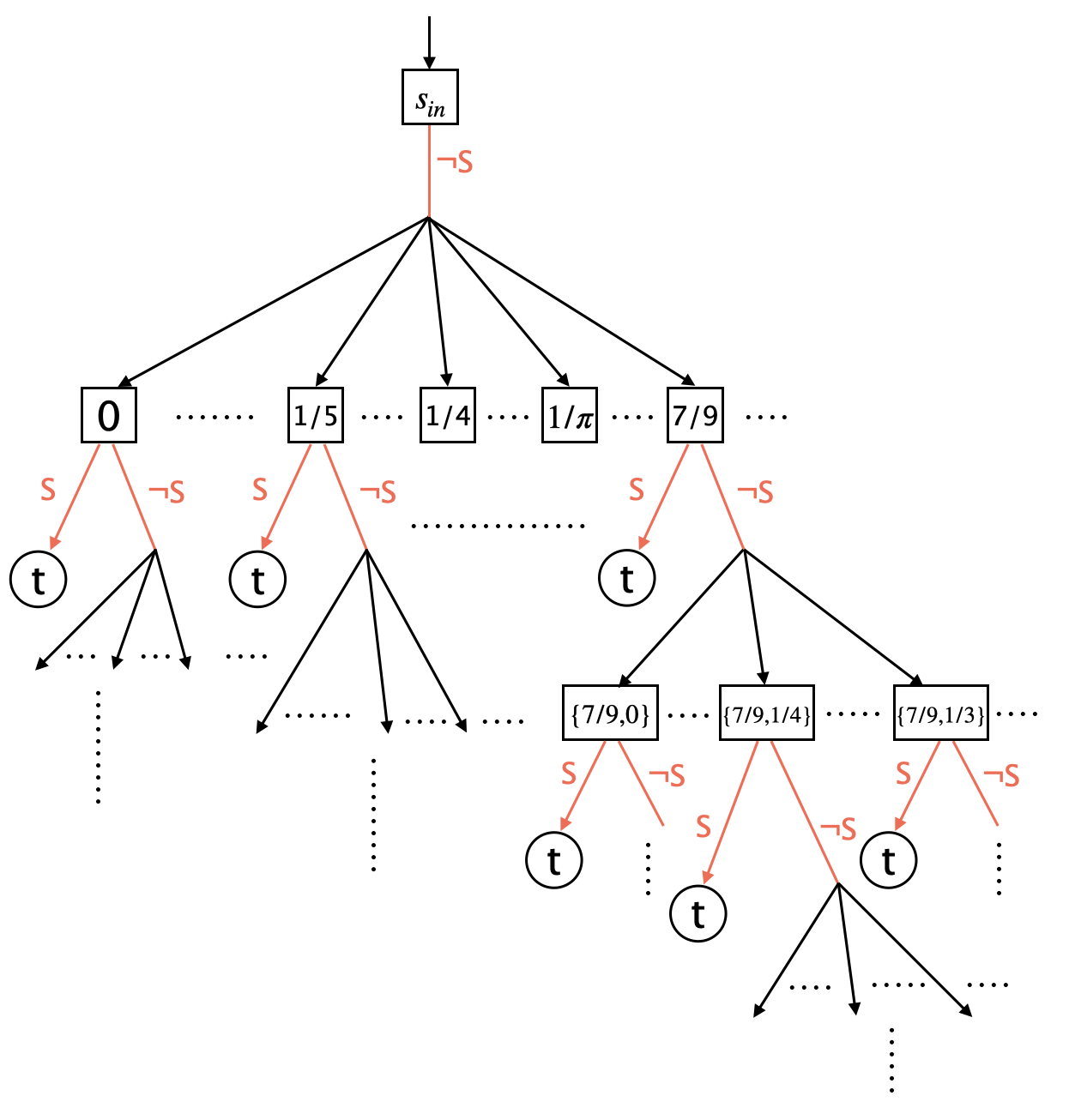}
      \caption{A part of a continuous MDP for Robbins' problem.}
      \label{fig:cont-MDP}
  \end{subfigure}
  \begin{subfigure}{.45\textwidth}
  \centering
      \includegraphics[scale=0.25]{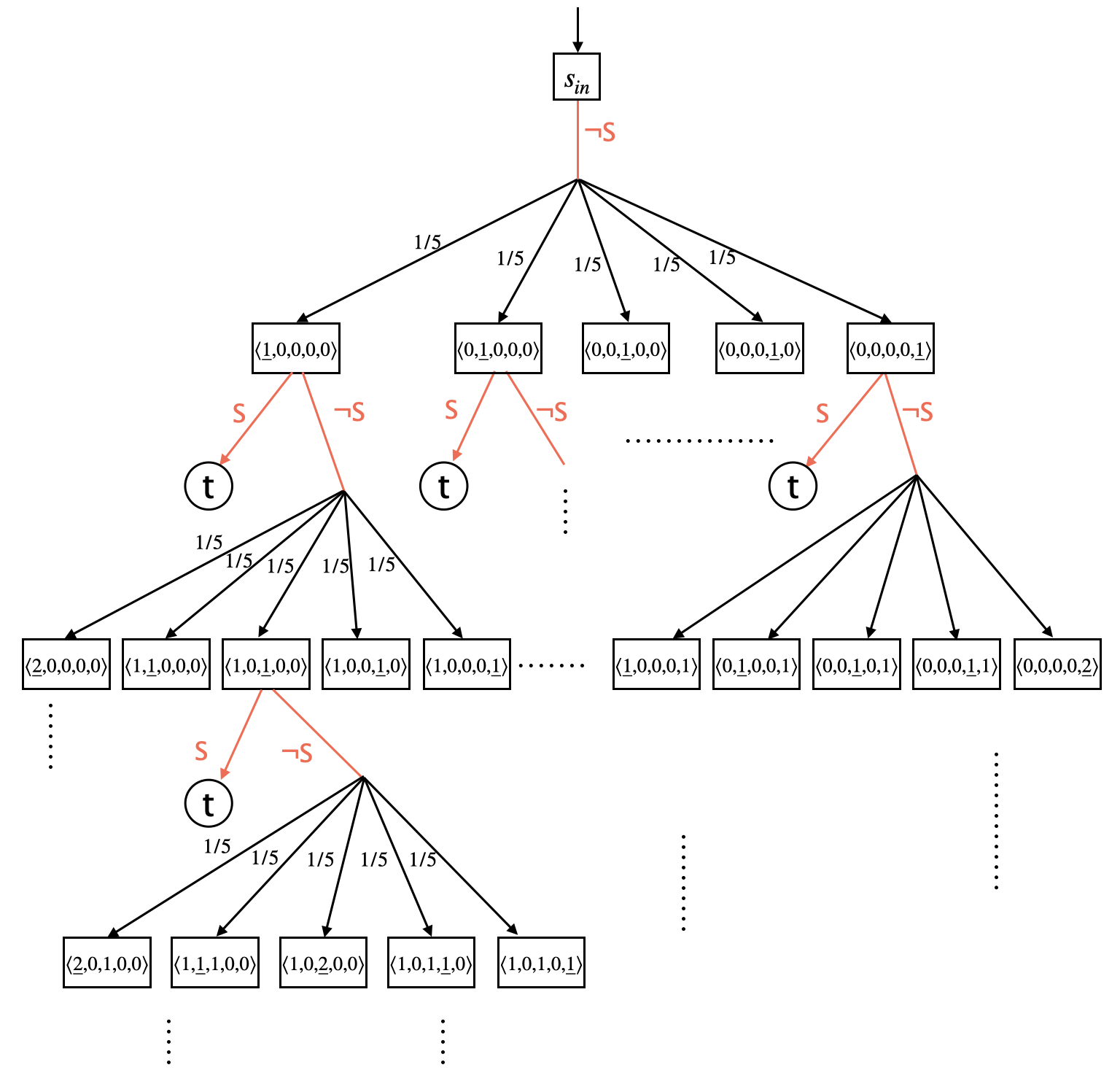}
      \caption{A part of a $d$-discrete MDP for Robbins' problem with $d = 5$.}
      \label{fig:d-abs-mdp}
  \end{subfigure}
\caption{Illustration of two types of MDPs: on the left, a continuous one for the exact Robbins' problem, and on the right, a $5$-discrete MDP. In both the MDPs, $S$ (resp.,~$\lnot S$) represents the protagonist's action $\stop$ (resp.,~$\cont$). For readability, we kept multiple copies of the  final state $t$. In the states of the discrete MDP, the interval of the most recent draw "$\last$" is marked with an underline; and further, for readability, the number of remaining draws "$r$" has been omitted from the states.}
\label{fig:cont-and-abs-mdp}
\end{figure}

\begin{definition}
\label{def:full-history-mdp}
For every $n, d \in \N$, we define the \emph{MDP} $\M_{n,d} = (S,A,s_{in},P,\Rew)$ where:
\begin{itemize}
    \item 
$S = \{(r, \last, \langle y_0, \cdots, y_{d-1} \rangle)\mid 0 \le r \le n, ~0 \le \last \le d-1, y_i \ge 0, ~\sum_{i=0}^{d-1}{y_i} = n-r\} \uplus \{s_{in}, t\}$, with $t$ being a final state;
\item $A = \{\cont, \stop \}$;
\item $s_{in} = (n,\vv{0})$ is the initial state, where $\vv{0}$ is the vector with all entries $0$;
\item 
\begin{itemize}
    \item 
for every $s = (n-1, m, \vv{y}) \in S\setminus \{s_{in}, t\}$, such that $\vv{y}[m] = 1$ and for all $j \neq m$, $\vv{y}[j] = 0$: 
$$P(s_{in},\cont,s)=1/d;$$
\item 
for every $s = (r, \last, \vv{y}) \in S\setminus \{s_{in}, t\}$ with $r > 0$, and for every $0 \le m \le d-1$,
$$P(s,\cont,s')=1/d,$$
where $s' = (r-1, m, \vv{y'}) \in S\setminus \{s_{in}, t\}$, such that $\vv{y'}[m] = \vv{y}[m] + 1$, and for all $j \neq m$, $\vv{y'}[j] = \vv{y}[j]$;
\item for every $s \in S\setminus \{s_{in}, t\}$: $P(s,\stop,t)=1$;
\item all other probabilities are $0$;
\end{itemize}

\item 
for every $s = (r, m, \vv{y})  \in S \setminus \{s_{in}, t\}:$
\begin{align}\label{eq:exp-rank-stop-full-hist}
    \Rew(s, \stop) = \sum_{i=0}^{m-1} \vv{y}[i] + \frac{\vv{y}[m] +1 }{2} + r\cdot\frac{2m+1}{2d}
\end{align}
\end{itemize}
\end{definition}

Intuitively, the first expression above counts the number of draws that are strictly before the current one, the second expression is the average of the number of draws in the same interval as the last draw, and the last expression is the expected number of draws that will be smaller among the remaining draws. Below, Lemma~\ref{lem:correct-full-hist} formalizes the correctness. 
Notice that, such an MDP is a directed acyclic graph. A part of this MDP with $d = 5$ is shown in Figure~\ref{fig:d-abs-mdp}.

The \emph{objective} then is to minimize the expected total loss in $\M_{n,d}$.

\begin{restatable}[Correctness]{lemma}{dcorr}
\label{lem:correct-full-hist}
    Fix $n, d$. Let $x_1, \cdots, x_{n-r}$ be the observations seen so far as in the description of Robbins' problem. Let $\vv{y}$ be the $d$-discretization of the sequence, \emph{i.e.}, $\vv{y} = \alpha_d(x_1, \cdots, x_{n-r})$ and let $x_{n-r} \in [m/d, (m+1)/d)$. 
    Let $\M_{n,d}$ be the MDP as defined in Definition~\ref{def:full-history-mdp} and let $s = (r,m,\vv{y})$.
    Then the conditional expected rank of $X_{n-r}$ given history $(\vv{y},m)$ is:
    $\E(L_{n-r}|(\vv{y},m)) = \Rew(s,\stop)$.
\end{restatable}

\begin{proof}
    \begin{align}
    \label{eq:correctness}
        \E(L_{n-r}|(\vv{y},m)) = 1 & + \sum\limits_{j<{n-r}}\Pr(X_j<X_{n-r}\mid (\vv{y},m))\nonumber \\
        & + \sum\limits_{j>{n-r}}\Pr(X_j<X_{n-r}\mid (\vv{y},m))
    \end{align}
    Now, for $j<{n-r}$,
    \begin{align}
    \label{eq:corr-1}
    \sum\limits_{j<{n-r}}\Pr(X_j<X_{n-r}\mid (\vv{y},m))   
     = \sum_{i=0}^{m-1} \vv{y}[i] +
      \frac{\vv{y}[m] - 1 }{2}
    \end{align}
      And, for $j>{n-r}$,
      \begin{align}
      \label{eq:corr-2}
      \Pr(X_j<X_{n-r}\mid (\vv{y},m)) & = \Pr(X_j<X_{n-r}\mid X_{n-r} \in [m/d, (m+1)/d)) \nonumber\\
      & = \int\limits_{m/d}^{(m+1)/d} \Pr(X_j < u)\cdot \frac{1}{1/d} du\nonumber\\
      & = d\cdot \left.\frac{u^2}{2} \right|_{m/d}^{(m+1)/d}\nonumber\\
      & = \frac{2m+1}{2d}
      \end{align}
      Thus, plugging Equations \eqref{eq:corr-1} and \eqref{eq:corr-2} into Equation \eqref{eq:correctness}, the lemma follows.
\end{proof}

The following two theorems characterize the quality of the approximation that we obtain with this discretization scheme. The proof of the first one is a direct consequence of Lemma~\ref{lem:correct-full-hist}, whereas the proof of the second theorem is more involved and deserves a detailed argument.

\begin{theorem}[Upper-bound]
\label{thm:approx}
For every $n$ and $d \in \mathbb{N}$, the minimum expected loss in $\M_{n,d}$ is at least $v_n$.
\end{theorem}

\begin{restatable}[Limit-completeness]{theorem}{LC}
\label{thm:limit-completness}
    For every $n$, as $d \to \infty$, the minimum expected loss in $\M_{n,d}$ converges to $v_n$. 
\end{restatable}

Let us briefly sketch a proof of Theorem~\ref{thm:limit-completness}. For a given $n \in \mathbb{N}$, both the exact and discretized Robbins' problems can be solved using backward induction. On complete branches of the tree of histories where no two draws fall in the same interval—let us call these cases \emph{conflict-free}—the value attributed to the branch in both processes (the exact and the discretized) is the same: the number of points preceding $x_n$ (the last draw) plus 1, and the number of occupied boxes before the box containing $x_n$ plus 1. This result generalizes to conflict-free branches where we decide to stop, if we limit the remaining draws to conflict-free draws: in that case, the expected value when we stop is close in both processes (details are given in the formal proof in appendix). In technical terms, the conditional expectation for conflict-free draws is close in both processes. Thus, both problems start with the same value in the leaves of the tree.  As $d$ tends to infinity, the measure of conflict-free branches approaches 1. This implies that as $d$ increases, the optimal losses in the two processes tend to the same value. Below, we present a formal proof of the theorem. 

\begin{proof}
\newcommand{\CF}{\mathsf{CF}}
    \textbf{Notations.}
    Let us fix $n \in \mathbb{N}$.
    Throughout this proof, for a history $hm$ (i.e., a state of the MDP $\M_n$), we write $\beta_d(h) = (r, m, \alpha_d(h))$ for the corresponding state in the MDP$\M_{n,d}$.
    For every MDP $\M$, every stationary strategy $\sigma$ in $\M$, and every state $s$ of $\M$, we write $\E(\sigma, s)$ for the expected loss of the player when she plays the strategy $\sigma$ from the state $s$ in $\M$.
    We write simply $\E(\sigma)$ when $s$ is the initial state of $\M$.
    On the other hand, we write $\E(\sigma, s, a)$ for that same expected payoff if her next action is $a$ instead of $\sigma(s)$.

    In the exact Robbins' problem MDP $\M_n$, let $\sigma^\star$ be an optimal strategy.
    In the discretized MDP $\M_{n,d}$, let $\sigma^d$ be an optimal strategy.
    By Theorem~\ref{thm:approx}, we have $\E\left(\sigma^\star\right) \leq \E\left(\sigma^d\right)$.
    In this proof, we will show that on the other hand, we have:
    $$\E\left(\sigma^d\right) \leq \E\left(\sigma^\star\right) + \left(1 - \left(\frac{d-n-1}{d}\right)^n\right) (n-1).$$
    Since the error term on the right converges to $0$ when $d$ tends to $+\infty$ (and $n$ is fixed), that will be sufficient to conclude.

    \textbf{Conflict-freeness.}
    Let us consider the case where draws are \emph{conflict-free}: for each $i$, we always have $\alpha_d(h)[i] \leq 1$.
    Conflict-freeness can be seen as an event, in the probabilistic sense of the term, in $\M_n$ as well as in $\M_{n,d}$.
    We write that event $\CF$.
    Let us however recall that in $\M_n$ as well as in $\M_{n, d}$, the player's loss when she stops on a state $s$ is defined by a formula that depends only on $s$, even though that formula has been designed, intuitively, as the expected rank of the current draw if the remaining draws are realized.
    Consequently, that formula still defines the loss when she stops on $s$ under the conflict-freeness hypothesis, which concerns only the values drawn before she stops.
    
    Let us define an intermediate strategy $\sigma$ in the discretized MDP $\M_{n,d}$.
    We define it so that it is optimal under conflict-freeness hypothesis, by backward induction:
    \begin{itemize}
        \item We define $\sigma(s) = \stop$ for every $s$ in which $n$ draws have already been done, or such that we have:
        $$\E\left(\sigma, s, \stop \mid \CF\right) \leq \E\left(\sigma, s, \neg\stop \mid \CF\right);$$

        \item and $\sigma(s) = \neg\stop$ for every other $s$.
    \end{itemize}

    Intuitively: the strategy $\sigma$ stops on the last draw, because it has to.
    At the previous draw, it stops if and only if the expected loss when stopping is better than or equal to the expected loss when not stopping.
    And so on, back to the first draw.
    By construction, the strategy $\sigma$ is optimal in $\M_{n,d}$ under the hypothesis that $\CF$ occurs; but without that hypothesis, we have $\E(\sigma^d) \leq \E(\sigma)$, by optimality of $\sigma^d$.
    To prove the result, we therefore need to prove the inequality:
    $$\E\left(\sigma\right) \leq \E\left(\sigma^\star\right) + \left(1 - \left(\frac{d-n-1}{d}\right)^n\right) (n-1).$$
    We will proceed by separating the conflict-free and the non-conflict free cases, using the equality:
    $$\E\left(\sigma\right) - \E\left(\sigma^\star\right)$$$$= \Pr(\CF) \left(\E\left(\sigma \mid \CF\right) - \E\left(\sigma^\star \mid \CF\right)\right) + \Pr(\overline \CF) \left(\E\left(\sigma \mid \overline{\CF}\right) - \E\left(\sigma^\star \mid \overline{\CF}\right)\right).$$
    
\textbf{Comparison between $\sigma$ and $\sigma^\star$.}
    We prove here that when conflict-freeness is guaranteed, the strategy $\sigma$ is actually at least as good as $\sigma^\star$.
    To do so, we prove by backward induction that $\E\left(\sigma, \beta_d(h) \mid \CF\right) \leq \E\left(\sigma^\star, h \mid \CF\right)$ for every conflict-free history $h$.

Let us first notice that, in the MDP $\M_{n,d}$, when the protagonist chooses to stop on the state $(r, m, \vv{y})$ with $\vv{y} = \alpha_d(h)$ such that $h$ is conflict-free, she immediately gets the loss:

\begin{align*}
    \E\left(\sigma, \beta_d(h), \stop \mid \CF\right) &= \sum_{i=0}^{m-1} \vv{y}[i] + \frac{\vv{y}[m] +1 }{2} + r\frac{2m+1}{2d} \\
    &= \sum_{i=0}^{m-1} \vv{y}[i] + \frac{1 + 1}{2} + r\frac{2m+1}{2d} \\
    &= \E(\sigma^\star, h, \stop \mid \CF)
\end{align*}

This result should not be a surprise: it means that when the last draw is $m$, discretization does not impact the computation of the expected rank if the player decides to stop on $m$, if all the past draws were unambiguously smaller or larger than $m$.

It is true in particular for our base case, when $|h| = n$ and $r=0$, in which case the protagonist has to stop.

    Let us now consider the inductive case, when $|h| < n$, i.e. $r > 0$.
    Let $X \subseteq [0,1)$ be the set of real numbers that can still be drawn, i.e. that are not in an interval that has already been touched.
    Let us also write $M$ for the measure of the set $X$, i.e.:
    $$M = \int_{m \in X}\mathrm{d}m = 1 - \frac{|h|}{d}.$$
    There are four cases to consider:
    \begin{itemize}
        \item If $\sigma^\star(h) = \sigma(\beta_d(h)) = \stop$, then, as exposed above, we have:
        $$\E\left(\sigma, \beta_d(h) \mid \CF\right) = \E\left(\sigma^\star, h\mid \CF\right).$$

        \item If $\sigma(\beta_d(h)) = \neg\stop$ and $\sigma^\star(h) = \stop$, then by definition of $\sigma$, continuing was better than or equivalent to stopping, under conflict-freeness hypothesis:
\begin{align*}
        \E\left(\sigma, \beta_d(h) \mid \CF\right) &\leq \E\left(\sigma, \beta_d(h), \stop \mid \CF\right) \\
        &= \E\left(\sigma^\star, h \mid \CF\right)
        \end{align*}

        \item If $\sigma(\beta_d(h)) = \stop$ and $\sigma^\star(h) = \neg\stop$, then by definition of $\sigma$, stopping was better than or equivalent to continuing, under conflict-freeness hypothesis.
        Using the induction hypothesis, we obtain:
        \begin{align*}
        \E\left(\sigma, \beta_d(h) \mid \CF\right) &\leq \E\left(\sigma, \beta_d(h), \neg\stop \mid \CF\right) \\
        &= \frac{1}{M} \int_{m \in X} \E\left(\sigma, \beta_d(h \cdot m) \mid \CF\right)\mathrm{d}m \\
        &\leq \frac{1}{M} \int_{m \in X} \E\left(\sigma^\star, h \cdot m \mid \CF\right)\mathrm{d}m \\
        &= \E\left(\sigma^\star, h\mid \CF\right)
        \end{align*}

        \item If $\sigma^\star(h) = \sigma(\beta_d(h)) = \neg\stop$, then similarly, using the induction hypothesis, we obtain:
        \begin{align*}
            \E\left(\sigma, \beta_d(h) \mid \CF\right)
            &= \frac{1}{M} \int_{m\in X} \E\left(\sigma, \beta_d(h \cdot m) \mid \CF\right) \mathrm{d}m  \\
            &\leq \frac{1}{M} \int_{m\in X} \E\left(\sigma^\star, h\cdot m \mid \CF\right)  \mathrm{d}m\\
            &= \E\left(\sigma^\star, h \mid \CF\right)
        \end{align*}

    \end{itemize}
    
    Therefore, in particular, under the hypothesis that the draws are conflict-free, we have $\E\left(\sigma \mid \CF\right) \leq \E\left(\sigma^\star \mid \CF\right)$.

\textbf{Conclusion.}
    Now, the probability of the event $\CF$ is:
    $$\Pr(\CF) = \frac{d \times (d-1) \times \dots (d-n)}{d^n} \geq \left(\frac{d-n}{n}\right)^n.$$

    When conflicts occur, we still have $\E\left(\sigma \mid \overline{\CF}\right) \leq \E\left(\sigma^\star \mid \overline{\CF}\right) + n-1$, since the possible ranks are $1, \dots, n$.
    Therefore, without conflict-freeness hypothesis, we have:
    \begin{align*}
        &\E\left(\sigma^d\right) - \E\left(\sigma^\star\right)\\
        &\leq \E\left(\sigma\right) - \E\left(\sigma^\star\right)\\
        &= \Pr(\CF) \left(\E\left(\sigma \mid \CF\right) - \E\left(\sigma^\star \mid \CF\right)\right) + \Pr(\overline \CF) \left(\E\left(\sigma \mid \overline{\CF}\right) - \E\left(\sigma^\star \mid \overline{\CF}\right)\right) \\
        &\leq 0 + \left(1 - \left(\frac{d-n-1}{d}\right)^n\right) (n-1).
    \end{align*}
    We find the inequality:
    $$\E\left(\sigma^\star\right) \leq \E\left(\sigma^d\right) \leq \E\left(\sigma^\star\right) + \left(1 - \left(\frac{d-n-1}{d}\right)^n\right) (n-1),$$
    which enables us to conclude to $\lim_{d\to \infty} \E\left(\sigma^d\right) = \E\left(\sigma^\star\right)$.
\end{proof}

\section{Remembering the $k$ best candidates only}
\label{sec:k-abs}

The abstraction in Definition~\ref{def:full-history-mdp} keeps track of the full history of the observations seen so far. The number of states in this MDP grows roughly ${\cal O}(n \cdot d^n)$, and so it limits our ability to compute the value of the MDP to small $n$ only. To be able to approximate the value $v_n$ for larger values of $n$, we consider another abstraction that consists in remembering only the counts for the $k$ smallest draws, \emph{i.e.}, the $k$ best candidates. In that case, the state space of the MDP grows only polynomially in $n$ -- the number of states is bounded by ${\cal O}(n \cdot d^k\cdot d) = {\cal O}(n \cdot d^{k+1})$. This second MDP abstraction is defined as follows.

\begin{definition}
\label{def:k-history-mdp}
For every $n,d$, and every $1 \le k \le n$, we define the \emph{MDP} $\M_{n,d,k} = (S,A,s_{in},P,\Rew)$ where
\begin{itemize}
    \item 
$S = \{(r, \last, \langle y_0, \cdots, y_{d-1} \rangle)\mid 0 \le r \le n, 0 \le \last \le d-1, y_i \ge 0, \sum_{i=0}^{d-1}{y_i} = \min(k,n-r)\} \uplus \{s_{in}, t\}$, with $t$ being a final state;
\item $A = \{\cont, \stop \}$;
\item $s_{in} = (n,\vv{0})$ is the initial state;
\item 
\begin{itemize}
    \item 
for every $s = (n-1, m, \vv{y}) \in S\setminus \{s_{in}, t\}$, such that $\vv{y}[m] = 1$ and for all $j \neq m$, $\vv{y}[j] = 0$: 
$$P(s_{in},\cont,s)=1/d;$$
\item 
for every $s = (r, \last, \vv{y}) \in S\setminus \{s_{in}, t\}$ with $n-r < k$, and for every $0 \le m \le d-1$, $$P(s,\cont,s')=1/d,$$ where $s' = (r-1, m, \vv{y'}) \in S\setminus \{s_{in}, t\}$, such that $\vv{y'}[m] = \vv{y}[m] + 1$ and for all $j \neq m$, $\vv{y'}[j] = \vv{y}[j]$;  
\item 
for every $s = (r, \last, \vv{y}) \in S\setminus \{s_{in}, t\}$ with $n-r \ge k$,  let $0 \le \maxi \le d-1$ be the \emph{maximum} index s.t. $\vv{y}[\maxi] \neq 0$. Then for every $0 \le m \le d-1$,
$$P(s,\cont,s')=1/d,$$ where
$s' = (r-1, m, \vv{y'})$ with 
$\vv{y'}[m] = \vv{y}[m] + 1
\text{ and }
\vv{y'}[\maxi] = \vv{y}[\maxi] - 1$;
\item for every $s \in S\setminus \{s_{in}, t\}$: $P(s,\stop,t)=1$;
\item all other probabilities are $0$;
\end{itemize}

\item 
for every $s = (r, m, \vv{y})  \in S \setminus \{s_{in}, t\}:$ with $0 \le \maxi \le d-1$ being the \emph{maximum} index s.t. $\vv{y}[\maxi] \neq 0$, the loss $\Rew(s, \stop)$ is defined as follows:
\begin{align}
\label{eq:exp-rank-stop-k-hist}
\begin{cases}
    \sum_{i=0}^{m-1} \vv{y}[i]+\frac{\vv{y}[m] +1 }{2} + r\cdot\frac{2m+1}{2d}
&\text{ if } m < \maxi \\
    \sum_{i=0}^{m-1} \vv{y}[i]+\frac{\vv{y}[m]+1}{2}+\max(0,\frac{n-r-k}{2(d-m)})+ r\cdot\frac{2m+1}{2d}
&\text{ if } m = \maxi\\
    1+k+\max(0,(n-r-k-1)\cdot\frac{2m-2\maxi+1}{2(d-\maxi)})+r\cdot\frac{2m+1}{2d}
    &\text{ if } m > \maxi
\end{cases}
\end{align}
\end{itemize}
\end{definition}
The \emph{objective} then is to minimize the expected total loss in $\M_{n,d,k}$.

\begin{lemma}
    For any $n, d$, the optimal loss of $\M_{n,d}$ of Definition~\ref{def:full-history-mdp} is the same as the optimal loss of $\M_{n,d,n}$ of Definition~\ref{def:k-history-mdp}.
\end{lemma}

We intuitively explain the loss function defined above. Here, in the first case, we have perfect information, and hence the loss is the same as in Definition~\ref{def:full-history-mdp}. In the second case, we may have forgotten some history that had fallen into the same interval as the last draw, and to capture this, we add the expression $\frac{n-r-k}{2(d-m)}$ to the loss. However, notice that, this quantity is only added when we already have stored $k$ histories, \emph{i.e.}, when $n-r > k$. Similarly, in the last case, we add $(n-r-k-1).\frac{2m-2\maxi+1}{2(d-\maxi)}$ to count for the draws that may have fallen between $\maxi$ and the current interval. Lemma~\ref{lem:correct-k-mdp} formalizes this intuition; the proof can be found in Appendix~\ref{app:lemma3-proof}.

\begin{restatable}[Correctness]{lemma}{kcorr}
\label{lem:correct-k-mdp}
    Fix $n, d$ and $1 \le k \le n$. Let $x_1, \cdots, x_{n-r}$ be the observations seen so far as in the description of Robbins' problem. Let $\vv{y}$ be the $d$-discretization of the sequence, \emph{i.e.}, $\vv{y} = \alpha_d(x_1, \cdots, x_{n-r})$ and let $x_{n-r} \in [m/d, (m+1)/d)$. Let $\M_{n,d,k}$ be the MDP as defined in Definition~\ref{def:k-history-mdp} and let $s = (r,m,\vv{y})$. Then,
    $\E(L_{n-r}|(\vv{y},m)) = \Rew(s,\stop)$.
\end{restatable}

\section{Remembering the $k$ best candidates in first $l$ intervals}
\label{sec:(k,l)-abs}

In the finite-state MDP that keeps track of the $k$-best draws seen so far, defined in Section~\ref{sec:k-abs}, the size grows as ${\cal O}(n \cdot d^{k+1})$, making it difficult to consider large values for $d$. Now, intuitively (as well as experimentally, cf. Table~\ref{tab:experiment}), as $n$ grows, $d$ should also be reasonably large to approximate $v_n$ sufficiently well (otherwise, the probability of collisions in histories is high and this is when we lose precision), and thus, computing the value for the MDP $\M_{n,d,k}$ may be infeasible.

However, note that if we remember the $k$ best draws, after a few rounds, it is likely that the values we have in memory are "good candidates". Let us say here that "good" means in $[0,l)$. For instance, if we remember the $k=3$ best draws, and we have seen $j$ draws taken uniformly at random, the probability that among the three best draws, one is above $l$, is bounded above by $(1-l)^{j-2}+(1-l)^{j-1}+(1-l)^{j}$. As a consequence, when $n$ is large, it is intuitively reasonable to only remember the exact interval of only the "good candidates" seen so far. Therefore, our next abstraction will consider a non-uniform partitioning of $[0,1)$, with a precise partitioning for $[0,l)$ into intervals of size $\frac{1}{d}$ as before, and only one interval for all the candidates that fall into $[l,1)$. This significantly reduces the number of states of the MDP, which is ${\cal O}(n \cdot l^{k+1})$. In Section~\ref{sec:experiments}, we will report on the values for several choices of $l$, and show that, in practice, as $n$ grows, a relatively smaller $l$ suffices for a "reasonable" approximation of the value $v_n$. The $(d,k,l)$-abstraction is formally defined as follows.

\begin{definition}
\label{def:d-k-l-abs}
For every $n,d$, every $1 \le k \le n$, and every $0 \le l \le d-1$ ,we define the \emph{MDP} $\M_{n,d,k,l} = (S,A,s_{in},P,\Rew)$ where
\begin{itemize}
    \item 
$S = \{(r, \last, \langle y_0, \cdots, y_{l-1}, y_l\rangle)\mid 0 \le r \le n, 0 \le \last \le l, y_i \ge 0, \sum\limits_{i=0}^{l}{y_i} = \min(k,n-r)\} \uplus \{s_{in}, t\}$, with $t$ being a final state;
\item $A = \{\cont, \stop \}$;
\item $s_{in} = (n,\vv{0})$ is the initial state;
\item 
\begin{itemize}
    \item 
for every $s = (n-1, m, \vv{y}) \in S\setminus \{s_{in}, t\}$, such that $\vv{y}[m] = 1$ and for all $j \neq m$, $\vv{y}[j] = 0$: 
\begin{align*}
P(s_{in},\cont,s)=
\begin{cases}
1/d &\text{ if } m < l \\
(d-l)/d &\text{ if } m = l 
\end{cases}
\end{align*}
\item 
for every $s = (r, \last, \vv{y}) \in S\setminus \{s_{in}, t\}$ with $n-r < k$, and for every $0 \le m \le l$, 
\begin{align*}
P(s,\cont,s')=
\begin{cases}
1/d &\text{ if } m < l \\
(d-l)/d &\text{ if } m = l 
\end{cases}
\end{align*}
where $s' = (r-1, m, \vv{y'}) \in S\setminus \{s_{in}, t\}$, such that $\vv{y'}[m] = \vv{y}[m] + 1$ and for all $j \neq m$, $\vv{y'}[j] = \vv{y}[j]$;
\item
for every $s = (r, \last, \vv{y}) \in S\setminus \{s_{in}, t\}$ with $n-r \ge k$,  let $0 \le \maxi \le l$ be the \emph{maximum} index s.t. $\vv{y}[\maxi] \neq 0$. Then for every $0 \le m \le l$,
\begin{align*}
P(s,\cont,s')=
\begin{cases}
1/d &\text{ if } m < l \\
(d-l)/d &\text{ if } m = l 
\end{cases}
\end{align*}
 where $s' = (r-1, m, \vv{y'})$ with 
$\vv{y'}[m] = \vv{y}[m] + 1
\text{ and }
\vv{y'}[\maxi] = \vv{y}[\maxi] - 1$;
\item for every $s \in S\setminus \{s_{in}, t\}$: $P(s,\stop,t)=1$;
\item all other probabilities are $0$;
\end{itemize}

\item
for every $s = (r, m, \vv{y})  \in S \setminus \{s_{in}, t\}:$ with $0 \le \maxi \le l$ the being \emph{maximum} index s.t. $\vv{y}[\maxi] \neq 0$, and $0 \le m \le l$, the loss $\Rew(s, \stop)$ is defined as follows:
\begin{itemize}
    \item[(i.)] if $\maxi < l$ and $m < l$, then $\Rew(s, \stop)$ is (same as Equation~\eqref{eq:exp-rank-stop-k-hist}):
\begin{align}
\label{eq:k-l-abs-rew1}
\begin{cases}
    \sum_{i=0}^{m-1} \vv{y}[i]+\frac{\vv{y}[m] +1 }{2} + r\cdot\frac{2m+1}{2d}
&\text{ if } m < \maxi \\
    \sum_{i=0}^{m-1} \vv{y}[i]+\frac{\vv{y}[m]+1}{2}+\max(0,\frac{n-r-k}{2(d-m)})+ r\cdot\frac{2m+1}{2d}
&\text{ if } m = \maxi\\
    1+k+\max(0,(n-r-k-1)\cdot\frac{2m-2\maxi+1}{2(d-\maxi)})+r\cdot\frac{2m+1}{2d}
&\text{ if } m > \maxi
\end{cases}
\end{align}
    \item [(ii.)]
    if $m = l$, then $\Rew(s, \stop)$ is:
    \begin{align}
\label{eq:k-l-abs-rew2}
\begin{cases}
(d-l)\cdot(\sum_{i=0}^{m-1} \vv{y}[i]+\frac{\vv{y}[m]+1}{2}+\max(0,\frac{n-r-k}{2(d-m)})+ r\cdot\frac{d+l}{2d})
&\text{ if } m = \maxi\\
    (d-l)\cdot(1+k+\max(0,(\text{n-r-k-}1)\cdot\frac{l+d-2.\maxi}{2(d-\maxi)})+r\cdot\frac{d+l}{2d})
&\text{ if } m > \maxi \\
\end{cases}
\end{align}
\end{itemize}
\end{itemize}
\end{definition}

The \emph{objective} then is to minimize the expected total loss in $\M_{n,d,k}$.

\begin{lemma}
    For every $n, d$ and every $1 \le k \le n$, the optimal loss of $\M_{n,d,k}$ of Definition~\ref{def:k-history-mdp} is the same as the optimal loss of $\M_{n,d,k,d-1}$ of Definition~\ref{def:d-k-l-abs}.
\end{lemma}

We can prove the correctness of loss function as in the previous abstractions:
\begin{lemma}[Correctness]
\label{lem:correct-d-k-l-abs}
    Fix $n, d, 1 \le k \le n$, and $0 \le l \le d-1$. Let $x_1, \cdots, x_{n-r}$ be the observations seen so far as in the description of Robbins' problem. Let $\vv{y}$ be the $d$-discretization of the sequence, \emph{i.e.}, $\vv{y} = \alpha_d(x_1, \cdots, x_{n-r})$ and let $x_{n-r} \in [m/d, (m+1)/d)$. 
    Let $\M_{n,d,k,l}$ be the MDP as in Definition~\ref{def:d-k-l-abs} and let $s = (r,m,\vv{y})$.
    Then,
    $\E(L_{n-r}|(\vv{y},m)) = \Rew(s,\stop)$.
\end{lemma}

\section{Experimental results}
\label{sec:experiments}

Our experimental results have been obtained using an implementation of the backward induction algorithm in {\sc Python3}\footnote{The source codes related to our algorithm can be found at \url{https://github.com/anirban11/Robbins_MDP}}.
In principle, we could have modelled our abstractions directly in tools like {\sc Prism}~\cite{prism} or {\sc Storm}~\cite{storm2017}, but the underlying MDPs are too large. Indeed, as we will report below, some of the MDPs we analysed have size (number of states plus number of edges) more than $10^{12}$. Consequently, the running time for some of the largest MDPs in our ad hoc implementation ranges from several hours to a few days.

\subsection{Experiments on $d$-abstractions of Section~\ref{sec:d-abs}}
We have been able to achieve almost optimal values for small values of $n$ with appropriate choices of $d$. For instance, it was shown in~\cite{assaf1996secretary}, that the optimal value for $n = 3$ is $v_3 = 1.3915\cdots$; and with $d = 500$ and $d = 1000$, we have achieved the values $v_{3,500} = 1.391635988$ and $v_{3,1000} = 1.391593893$. Similarly, in~\cite{n=4}, the authors show that $v_4 = 1.4932\cdots$ (the method being reasonably complex); and we could achieve close values with, for example, $d= 500$ and $d = 1000$: $v_{4,500} = 1.493418067584$ and $v_{4,1000} = 1.493356615167$.

We conclude that our first abstraction is capable of computing (near) optimal rank for any $n$ (number of draws), with appropriate choices for $d$. However, as noted earlier,  the size of this MDP is roughly ${\cal O}(n\cdot d^n)$, which is a concern, since for larger $n$, the choice of $d$ also should be `reasonably' large (depending on $n$) for obtaining a `reasonable' approximation of the optimal rank.
It becomes almost infeasible to compute the value of this abstraction for $n \ge 6$ , even with precision $d=100$. In the following, we will see that our second abstraction is more tractable as well as can achieve `good' approximations for larger $n$.

\subsection{Experiments on $(d,k)$-abstractions of Section~\ref{sec:k-abs}}

	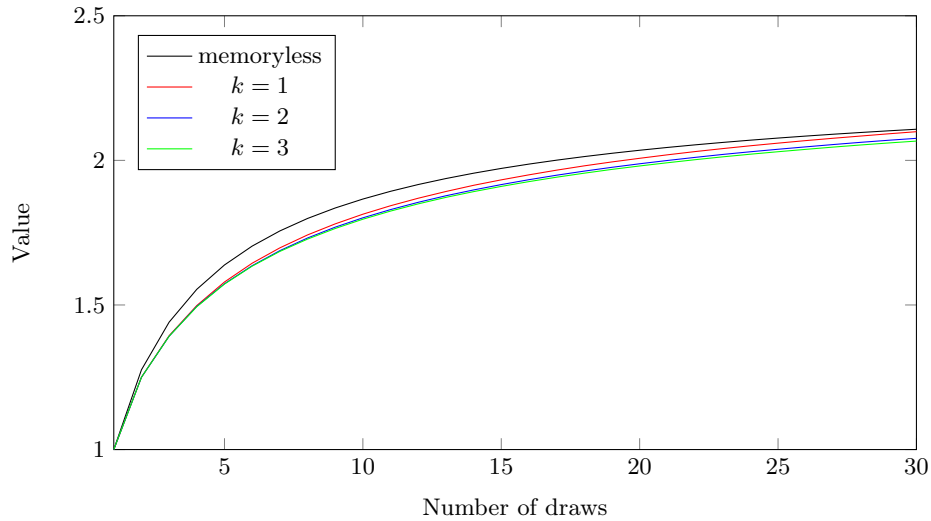
\begin{figure}[t]
	    \centering
	\begin{tikzpicture}
		\pgfplotsset{%
			width=1\textwidth,
			height=0.6\textwidth,
   legend style = {at={(0.03,0.8)}, anchor = west}
		}
		\pgfplotstableread{
			{$n$} {memoryless} {$f(n,d=100,k=1)$} {$f(n,d=100,k=2)$} {$f(n,d=100,k=3)$}		
1 1.0 1 1 1
2 1.275811749652425 1.25 1.25 1.25
3 1.4406553537397055 1.3937410297446928 1.3919754999999998 1.3919754999999998
4 1.554231487592793 1.4980773499999995 1.4952335448304033 1.49395909
5 1.6383330315011604 1.5795054408434996 1.5734700490994833 1.5721921488824606
6 1.7035992657834345 1.6444245663376638 1.6358619499403502 1.6342362984803323
7 1.755981689668285 1.6974387879871065 1.6877894583822046 1.685101537457593
8 1.7991101490119972 1.7421584996749382 1.7315605800759681 1.7276892448804642
9 1.8353399442411704 1.78038030836292 1.7690237443010937 1.764445000811847
10 1.8662733112166188 1.813875528063636 1.8013254918718997 1.796538612928518
11 1.8930418652242904 1.8431356999620616 1.8297817551886348 1.8246501557474404
12 1.9164696896925588 1.8691270817016201 1.8550394330846165 1.8496203931498636
13 1.9371724315418803 1.892395646472228 1.877658946759951 1.872054518892301
14 1.955620083731664 1.9137717970997081 1.8977472994465094 1.8919120889360712
15 1.9721781798256033 1.9326420419266084 1.9162765482680093 1.9099663016179091
16 1.9871356427648257 1.9499738594934442 1.9331822317861294 1.9265730076497491
17 2.000724103173354 1.966190259446156 1.948711934623551 1.9418095400830386
18 2.0131316053179527 1.981105091777348 1.9628701530999633 1.9557362582774984
19 2.0245125261168546 1.9946405523612243 1.9760182987136068 1.9688293510947446
20 2.0349948812329974 2.007183109135607 1.9882264479394598 1.98068709413836
21 2.0446857921469492 2.019166566076766 1.999598539145068 1.991764623401514
22 2.0536756356449932 2.030373747812136 2.010320306208866 2.002311699407846
23 2.062041234034737 2.0407173263949723 2.020492045374007 2.0122038667799553
24 2.069848336720526 2.0504748799336916 2.0297647718790834 2.021338458645636
25 2.0771535713006974 2.0596020770574235 2.0384359020983833 2.029912255132422
26 2.0840059927096384 2.068350360925118 2.046745839201283 2.038065694244965
27 2.0904483243764487 2.0765323809789775 2.0547954781536433 2.0459936472399267
28 2.096517960965898 2.0842570627155808 2.0624857989241203 2.0533472446013734
29 2.102247784792144 2.091923487641113 2.0696130629442555 2.0602216687103643
30 2.107666835325425 2.0991683008892688 2.076211792495852 2.066722623107619
		}\mytable

    \begin{axis}[
			xmin=1, xmax=30,
	ymin=1, ymax=2.5,
	ylabel={Value},
	ylabel style={align=center},
	xlabel={Number of draws},
	mark=none
	]
	
	\addplot [black, mark = none] table[x index=0,y index=1] {\mytable};
	\addlegendentry{memoryless}
	\addplot [red, mark = none] table[x index=0,y index=2] {\mytable};
	\addlegendentry{$k=1$}
	\addplot [blue, mark = none] table[x index=0,y index=3] {\mytable};
	\addlegendentry{$k=2$}
	\addplot [green, mark = none] table[x index=0,y index=4] {\mytable};
	\addlegendentry{$k=3$}
    \end{axis}
    \end{tikzpicture}
    \caption{Values for $1 \le n\le 30$ with the memoryless strategy of~\cite{assaf1996secretary} with the parameter $c=1.9469$; and for our $(d,k)$-abstraction -- with the parameters $d=100$ and $k = 1,2,3$ -- are presented.}
    \label{fig:experiment}
\end{figure}

Consider the case $n=4$ as discussed in the previous subsection. We observe that the optimal value for $n=4$ is less than $2$. This intuitively implies we can get a `good' approximation even when we only store the $2$ best draws seen so far. For instance, $v_{4,500,2} = 1.49433$ and $v_{4,1000,2} = 1.49423$. 
Based on this intuition, using $(d,k)$-abstraction, with smaller values of $k$ (e.g., $k \in \{2,3\}$), we can achieve values that are close to the optimal expected ranks for larger values of $n$. In the following, we first compare the optimal values for $n \le 30$ with different values of $k$ (w.r.t. $d = 100$) (cf. Figure~\ref{fig:experiment}), and second, we report on the values for larger $n$ with $k \in \{2,3\}$ (cf. Table~\ref{tab:experiment}) and will see that, in most of the cases (for $n\le 100$), we can achieve better values than the optimal values obtained by the memoryless strategy of~\cite{assaf1996secretary}. 

Having said the above, Figure~\ref{fig:experiment} shows a graphical comparison of the minimal expected ranks computed by the memoryless strategy of~\cite{assaf1996secretary}, with the $(d,k)$-abstraction, defined in Section~\ref{sec:k-abs}, for the number of draws $n \le 30$, with the precision $d =100$, and for $k \in \{1,2,3\}$. Notice that, with this precision, we always achieve a better value than the memoryless one of~\cite{assaf1996secretary}, for all $k \in \{1,2,3\}$.

As mentioned earlier, we would like to clarify that the strategies considered in~\cite{assaf1996secretary} are not truly "memoryless" in the classical sense. Rather, they are "counting" strategies, meaning they store the information about how many draws remain in the state space, denoted as $r$. Additionally, from this information, they infer that the $i^{th}$ draw (for all $i \le n-r$) was in the interval $[\varphi(n,i),1)$, where $\varphi$ is the threshold function defined in Section~\ref{sec:intro}.

\begin{table}[t]
\caption{Minimal expected rank for several choices of $n$, for $d = 100,500$ with $k = 2,3$; and for $d = 1000$ with $k = 2$ are presented. Additionally, we also report on the values achieved by the memoryless strategy of~\cite{assaf1996secretary} to highlight the fact that we get better values with our abstractions for most $n$, with appropriate $d, k$.}
\begin{center}
\begin{tabular}{||c||c|c||c|c||c||c||c||}
    \hline
     \multirow{2}{*}{$n$}& \multicolumn{2}{c||}{$d=100$} & \multicolumn{2}{c||}{$d=500$} & {$d=1000$} & \multirow{2}{*}{$ml$~\cite{assaf1996secretary}} \\
     \cline{2-6}
      &$k=2$ & $k=3$ &$k=2$ & $k=3$  &$k=2$ & \\
    \hline
    $5$ & 1.57347 & 1.57219 & 1.57231 & 1.57128 & 1.57217 &1.63833 \\
    \hline
    $10$ & 1.80133 & 1.79654 & 1.79912 & 1.79463 & 1.79886 &1.86627 \\
    \hline
    $50$ & 2.16781 & 2.15594 & 2.15293 & 2.14325 & 2.15155 &2.17663 \\
    \hline
    $100$ & 2.26683 & 2.25158 & 2.23382 & 2.22249 & 2.23077 &2.23940 \\
    \hline
    $500$ & 3.10808 & 3.10754 & 2.34799 & 2.33137 & 2.32697 &2.30574 \\
    \hline
\end{tabular}
\end{center}
\label{tab:experiment}
\end{table}

In Table~\ref{tab:experiment}, we present values given by our $(d,k)$-abstraction for different $n$'s and again compare with~\cite{assaf1996secretary}. For all $n\le 100$, we could achieve better values than~\cite{assaf1996secretary} with an appropriate choice of $d$ (e.g., $d=100,500$ or $1000$) and even with relatively small $k$ (e.g., $k \in \{2,3\}$). For $n=500$, with $d=1000,~ k=2$, we achieve a value close to that of~\cite{assaf1996secretary}. Notice that, as described earlier, as $n$ grows, $d$ should also be larger - for instance, for $n=500$, $d = 100$ is of not much help.
All values for $n = 1$ to $100$ with $d = 500$, $1000$ and $k = 2$, and for the memoryless strategy of~\cite{assaf1996secretary} are reported in Appendix~\ref{app:experiment}. From those values, we conclude that we achieve better values for all $n \le 100$ than the memoryless strategy of~\cite{assaf1996secretary}.
We would also like to mention that, for $d = 1000$ with $k = 3$, our system ran out of memory.

As mentioned earlier, some of the MDPs above are very large. For instance, the MDP $\M_{100,500,3}$ has size\footnote{The size here refers to the total number of states and transitions in the MDP.} roughly $6\times 10^{12}$, and the time taken to compute the value for this set of parameters was $1.6 \times 10^5$ seconds ($\approx 45$ hours). To take another example, the MDP $\M_{500,1000,2}$ has size around $5 \times 10^{11}$, and the computation time was approx. $1.3 \times 10^5$ seconds. This showcases the efficiency of our implementation.

\subsection{Experiments on $(d,k,l)$-abstractions of Section~\ref{sec:(k,l)-abs}}
The results of this section demonstrate, as described earlier, that for large $n$'s, all $k$-best draws fall in the interval $[0,l)$ for reasonable choices of $l$, e.g., $l = d/10, d/5$ etc. with high probability. Therefore, while keeping track of the best draws in $[0,l)$ gives a good estimate of the optimal rank, on the other hand, it is much faster than our previous abstraction.
Let us give an example. Consider $n = 500, k=2$. With $d = 500$,  the value is $v_{500,500,2} = 2.34799$ which took around $1.9 \times 10^4$ seconds, whereas $v_{500,2000,2,300} = 2.32791$, which is a better value than the previous, and the computation took only $2\times 10^3$ seconds. The latter is still a `good' over-approximation of $v_{500,1000,2}=2.32697$ which took $1.3 \times 10^5$ seconds to compute. 
With $k=3$, the computation of $v_{500,500,3}$ takes up to $11$ days.

\section{Conclusion}
\label{sec:conclusion}

In this paper, we have studied several abstractions for  Robbins' problem, modelling them as Markov Decision Processes. These abstractions have enabled us to compute approximate values for instances of  Robbins' problem where no good approximations were previously known, specifically for all values between $n = 5$ and $n = 100$ for which we systematically obtain values that are better than the best previously known upper-bounds.

More importantly, our findings demonstrate that the memory of past draws can be concentrated on the few best draws while forgetting the less favourable ones. Specifically, when we remember the two or three best previous draws, our strategies are almost optimal. This indicates that even though the optimal strategy theoretically depends on the entire history, the optimal value may be effectively approximated with a relatively small number of draws seen so far.

\paragraph*{Possible Future Works.}
First, it could be interesting to explore the problem of computing lower bounds numerically using MDPs, as this paper has primarily focused on methods for obtaining upper bounds. 
Second, another interesting question is to know whether our discretization approach may also give a new access to approximate the solution of a differential equation derived in~\cite{bruss2009}.
 In that article,  $n$ is thought to be the realization of a random variable $N_n$
 which is the count of i.i.d.  events in a Poisson process of rate $1$ on an interval $[0,n]$, $ n= 1, 2, \cdots$.   Hence, $\E(N)/N_n \to 1$ almost surely, and thus,  it
is somewhat intuitive that the corresponding optimal value should coincide in the limit with our searched value $v$ of the original Robbins’ problem.
This is indeed true, however, we should mention that the proof is more complicated than one would think,
and the complete proof is given in~\cite{bruss2009}.
With this result, we are able to argue with increments of relative ranks on intervals of the form $[t, t+dt)$ and to express the limiting value $v$ as the solution
 of a differential equation in two unknown functions, one of which  is shown to tend to $0$ as $t=n$ tends to $\infty.$
The better we can estimate the latter, the closer we come to $v.$ The difficulty of full history dependence therefore does  not disappear, but this setting allows for an alternative closed-form approach.

  \paragraph*{Acknowledgements.} We would like to thank Purba Das, lecturer in the Department of Mathematics at King’s College London, for helpful discussions. We would also like to thank the anonymous reviewers of this paper for their valuable feedback and suggestions.


%
\bibliographystyle{splncs04}
\bibliography{main}

@book{DBLP:books/daglib/0020348,
  author       = {Christel Baier and
                  Joost{-}Pieter Katoen},
  title        = {Principles of model checking},
  publisher    = {{MIT} Press},
  year         = {2008}
}

@book{DBLP:books/wi/Puterman94,
  author       = {Martin L. Puterman},
  title        = {Markov Decision Processes: Discrete Stochastic Dynamic Programming},
  series       = {Wiley Series in Probability and Statistics},
  publisher    = {Wiley},
  year         = {1994}
}

@article{bruss1993minimizing,
  title={Minimizing the expected rank with full information},
  author={Bruss, F Thomas and Ferguson, Thomas S},
  journal={Journal of Applied Probability},
  volume={30},
  number={3},
  pages={616--626},
  year={1993},
  publisher={Cambridge University Press}
}

@article{assaf1996secretary,
  title={The secretary problem: minimizing the expected rank with iid random variables},
  author={Assaf, David and Samuel-Cahn, Ester},
  journal={Advances in Applied Probability},
  volume={28},
  number={3},
  pages={828--852},
  year={1996},
  publisher={Cambridge University Press}
}

@article{n=4,
  title={One step more in {R}obbins' problem: Explicit solution for the case n = 4},
  author={Dendievel, R{\'e}mi and Swan, Y},
  journal={Mathematica Applicanda},
  volume={44},
  number={1},
  pages={135--148},
  year={2016}
}

@article{chow1964,
  title={Optimal selection based on relative rank (the “secretary problem”)},
  author={Chow, Yuan-Shih and Moriguti, Sigaiti and Robbins, Herbert and Samuels, Stephen Mitchell},
  journal={Israel Journal of {M}athematics},
  volume={2},
  number={2},
  pages={81--90},
  year={1964},
  publisher={Springer}
}

@article{bruss2005,
  title={What is known about {R}obbins' problem?},
  author={Bruss, F Thomas},
  journal={Journal of Applied Probability},
  volume={42},
  number={1},
  pages={108--120},
  year={2005},
  publisher={Cambridge University Press}
}

@inproceedings{storm2017,
  author       = {Christian Dehnert and
                  Sebastian Junges and
                  Joost{-}Pieter Katoen and
                  Matthias Volk},
  title        = {A Storm is Coming: {A} Modern Probabilistic Model Checker},
  booktitle    = {{CAV} {(2)}},
  series       = {Lecture Notes in Computer Science},
  volume       = {10427},
  pages        = {592--600},
  publisher    = {Springer},
  year         = {2017}
}

@inproceedings{prism,
  title={PRISM: Probabilistic symbolic model checker},
  author={Kwiatkowska, Marta and Norman, Gethin and Parker, David},
  booktitle={International Conference on Modelling Techniques and Tools for Computer Performance Evaluation},
  pages={200--204},
  year={2002},
  organization={Springer}
}

@article{gilbert1966,
  title={Recognizing the maximum of a sequence},
  author={Gilbert, John P and Mosteller, Frederick},
  journal={Journal of the American Statistical Association},
  volume={61},
  number={313},
  pages={35--73},
  year={1966},
  publisher={Taylor \& Francis}
}

@article{lindley1961,
  title={Dynamic programming and decision theory},
  author={Lindley, Denis V},
  journal={Journal of the Royal Statistical Society: Series C (Applied Statistics)},
  volume={10},
  number={1},
  pages={39--51},
  year={1961},
  publisher={Wiley Online Library}
}

@inproceedings{bruss1996,
  title={Half-prophets and {R}obbins’ problem of minimizing the expected rank},
  author={Bruss, F Thomas and Ferguson, Thomas S},
  booktitle={Athens Conference on Applied Probability and Time Series Analysis. Lecture Notes in Statistics, Volume 114},
  pages={1--17},
  year={1996},
  organization={Springer}
}

@article{meier2017,
  title={A new strategy for {R}obbins’ problem of optimal stopping},
  author={Meier, Martin and S{\"o}gner, Leopold},
  journal={Journal of Applied Probability},
  volume={54},
  number={1},
  pages={331--336},
  year={2017},
  publisher={Cambridge University Press}
}

@article{gnedin2007,
  title={Optimal stopping with rank-dependent loss},
  author={Gnedin, Alexander V},
  journal={Journal of Applied Probability},
  volume={44},
  number={4},
  pages={996--1011},
  year={2007},
  publisher={Cambridge University Press}
}

@article{bruss2009,
  title={A continuous-time approach to {R}obbins' problem of minimizing the expected rank},
  author={Bruss, F Thomas and Swan, Yvik C},
  journal={Journal of Applied Probability},
  volume={46},
  number={1},
  pages={1--18},
  year={2009},
  publisher={Cambridge University Press}
}

@article{bruss2024,
  title={Mathematical Intuition, Deep Learning, and {R}obbins’ Problem},
  author={Bruss, F Thomas},
  journal={Jahresbericht der Deutschen Mathematiker-Vereinigung},
  pages={1--25},
  year={2024},
  publisher={Springer}
}

\newpage
\appendix

\section{Appendix}
\label{sec:appendix}

\subsection{Proof of Lemma~\ref{lem:correct-k-mdp}}
\label{app:lemma3-proof}
\kcorr*

\begin{proof}
    \begin{align}
    \label{eq:correctness-1-k-mdp}
        \E(L_{n-r}|(\vv{y},m)) = 1 &+ \sum\limits_{j<{n-r}}\Pr(X_j<X_{n-r}\mid (\vv{y},m)) \nonumber \\
        &+ \sum\limits_{j>{n-r}}\Pr(X_j<X_{n-r}\mid (\vv{y},m))
    \end{align}
    Let $0 \le \maxi \le d-1$ be the \emph{maximum} index s.t. $\vv{y}[\maxi] \neq 0$. Then, if $m < \maxi$, we have perfect (w.r.t. $d$-discretization) history, and thus the proof is the same as for Lemma~\ref{lem:correct-full-hist}.
    It is also the same when ${n-r}<k$.
    
    When ${n-r} > k$, there are two possibilities:
    \begin{itemize}
        \item [(i.)]
    if $m= \maxi$, we may have forgotten some histories (w.r.t. $d$-discretization). Then, for $j < {n-r}$,
      \begin{align}
      \label{eq:correctness-2-k-mdp}
      &\Pr(X_j<X_{n-r}\mid (\vv{y},m))  \nonumber\\ & = 
      \sum_{i=0}^{m-1} \vv{y}[i] + \frac{\vv{y}[m] - 1 }{2}\nonumber\\
      & + (n-r-k)\cdot
      \Pr(X_j<X_{n-r}\mid X_{n-r} \in [m/d, (m+1)/d); X_j > m/d) 
      \end{align}

      The last probability above evaluates to:
      
      \begin{align}
      \label{eq:correctness-3-k-mdp}
      &=\int\limits_{m/d}^{(m+1)/d}\Pr(X_j<u \mid X_j \in [m/d, (m+1)/d))\cdot \frac{1}{1/d}\nonumber\\
      & =d\cdot\int\limits_{m/d}^{(m+1)/d} \frac{u-m/d}{1-m/d} du\nonumber\\
      & = \frac{1}{2(d-m)}.
      \end{align}

    \item [(ii.)]  If $m>\maxi$, again, we may have forgotten some histories. For $j<{n-r}$,
      \begin{align}
      \label{eq:correctness-4-k-mdp}
      &\Pr(X_j<X_{n-r}\mid (\vv{y},m)) \nonumber\\
      &= k+ ({n-r-k-}1)\cdot
      \Pr(X_j<X_{n-r}\mid X_{n-r} \in [m/d, (m+1)/d); X_j > \maxi/d) 
      \end{align}
      Similar to the previous case, 
      the last probability above evaluates to $\frac{2m-2\maxi+1}{2(d-\maxi)}$.
    \end{itemize}
    Lastly, for both (i.) and (ii.), one can show that for $j>{n-r}$, $$\Pr(X_j<X_{n-r}\mid (\vv{y},m)) = \frac{2m+1}{2d},$$ similar to the proof of Lemma~\ref{lem:correct-full-hist}.
    
    This concludes the proof of Lemma~\ref{lem:correct-k-mdp}.
\end{proof}

\subsection{Experimental Results}
\label{app:experiment}
Below, in Table~\ref{tab:full-table-1} and~\ref{tab:full-table-2}, we present the values for $n = 1$ to $100$ with $d = 500$, $1000$ and $k = 2$, and for the memoryless strategy of~\cite{assaf1996secretary} with the parameter $c=1.9469$.

\begin{table}[t]
\caption{Experimental results for $1 \le n \le 50$.}
\centering
\begin{tabular}{|c|c|c|c|}
\hline
$n$ & $d=500$, $k=2$ & $d=1000$, $k=2$ & $ml$~\cite{assaf1996secretary} \\
\hline
 1 & 1.0 & 1.0 & 1.0\\
\hline
 2 & 1.2499999999999998 & 1.2500000000000002 & 1.275811749652425\\
\hline
 3 & 1.3916359880000004 & 1.3915938930000002 & 1.4406553537397055\\
\hline
 4 & 1.4943354630156864 & 1.4942292766819876 & 1.554231487592793\\
\hline
 5 & 1.5723087548475956 & 1.5721659394269207 & 1.6383330315011604\\
\hline
 6 & 1.6344823513487399 & 1.6343105674656457 & 1.7035992657834345\\
\hline
 7 & 1.686172134541134 & 1.685981863167616 & 1.755981689668285\\
\hline
 8 & 1.7297836238112014 & 1.7295734860288026 & 1.7991101490119972\\
\hline
 9 & 1.76697396489574 & 1.7667368261168588 & 1.8353399442411704\\
\hline
 10 & 1.7991230945455245 & 1.7988592862297756 & 1.8662733112166188\\
\hline
 11 & 1.8272910253255565 & 1.8270040145651623 & 1.8930418652242904\\
\hline
 12 & 1.8522528957634636 & 1.8519378867449332 & 1.9164696896925588\\
\hline
 13 & 1.8745640590980146 & 1.874224452938298 & 1.9371724315418803\\
\hline
 14 & 1.8946618448918564 & 1.8942880399765794 & 1.955620083731664\\
\hline
 15 & 1.9128629244215691 & 1.912469614982583 & 1.9721781798256033\\
\hline
 16 & 1.9294478151453334 & 1.929033113555414 & 1.9871356427648257\\
\hline
 17 & 1.9446372244176986 & 1.9442008003839266 & 2.000724103173354\\
\hline
 18 & 1.9586227440120254 & 1.958150441973076 & 2.0131316053179527\\
\hline
 19 & 1.9715246546331406 & 1.971031895791245 & 2.0245125261168546\\
\hline
 20 & 1.983506256178762 & 1.9829720564818833 & 2.0349948812329974\\
\hline
 21 & 1.9946168335447256 & 1.9940766273998967 & 2.0446857921469492\\
\hline
 22 & 2.0050126252872933 & 2.004429007754717 & 2.0536756356449932\\
\hline
 23 & 2.0147173693490203 & 2.014113114384068 & 2.062041234034737\\
\hline
 24 & 2.023820142205404 & 2.0231888014277133 & 2.069848336720526\\
\hline
 25 & 2.0323804104347856 & 2.0317276626270973 & 2.0771535713006974\\
\hline
 26 & 2.0404460915509004 & 2.039760355623614 & 2.0840059927096384\\
\hline
 27 & 2.0480519382492663 & 2.047346308631722 & 2.0904483243764487\\
\hline
 28 & 2.055253315368793 & 2.0545165201468136 & 2.096517960965898\\
\hline
 29 & 2.06208187185613 & 2.0613137896711415 & 2.102247784792144\\
\hline
 30 & 2.0685719367088833 & 2.067757442307411 & 2.107666835325425\\
\hline
 31 & 2.074725546620756 & 2.073881232353291 & 2.1128008619181733\\
\hline
 32 & 2.080552090798597 & 2.079715772444592 & 2.1176727829861313\\
\hline
 33 & 2.0861324542984554 & 2.085268062249322 & 2.122303069719032\\
\hline
 34 & 2.0914836026504333 & 2.0905692353602703 & 2.126710068493654\\
\hline
 35 & 2.0965752795598043 & 2.0956356458969365 & 2.1309102731861964\\
\hline
 36 & 2.1014658624003406 & 2.1004766720179333 & 2.134918556292427\\
\hline
 37 & 2.1061101650196856 & 2.105115614550015 & 2.1387483659904847\\
\hline
 38 & 2.110596788039435 & 2.109565837909762 & 2.1424118948965996\\
\hline
 39 & 2.1148757196539156 & 2.1138301822128676 & 2.145920225175551\\
\hline
 40 & 2.1190202503144295 & 2.1179304022088483 & 2.1492834538064725\\
\hline
 41 & 2.1229810323778056 & 2.1218736663243685 & 2.152510801118974\\
\hline
 42 & 2.1268059866429283 & 2.1256706784889094 & 2.1556107051654525\\
\hline
 43 & 2.1304854284787837 & 2.1293242358747553 & 2.1585909040533426\\
\hline
 44 & 2.134032347521156 & 2.132833017064226 & 2.1614585080030713\\
\hline
 45 & 2.13747622844871 & 2.1362219755168406 & 2.1642200626062627\\
\hline
 46 & 2.1407624248039046 & 2.139503475176765 & 2.1668816045206576\\
\hline
 47 & 2.143947259608725 & 2.142680408893137 & 2.1694487106425955\\
\hline
 48 & 2.147055594572068 & 2.145737152434097 & 2.171926541636693\\
\hline
 49 & 2.1500380464592714 & 2.1486924123422817 & 2.174319880568634\\
\hline
 50 & 2.1529295415335197 & 2.151552133002007 & 2.176633167275941\\
\hline
\end{tabular}
\label{tab:full-table-1}
\end{table}
\begin{table}[t]
\caption{Experimental results for $51 \le n \le 100$.}
\centering
 \begin{tabular}{|c|c|c|c|}
\hline
$n$ & $d=500$, $k=2$ & $d=1000$, $k=2$ & $ml$~\cite{assaf1996secretary} \\
\hline
 51 & 2.1557396454186146 & 2.1543228280961673 & 2.1788705290187496\\
\hline
 52 & 2.1584559348776926 & 2.1570154817019893 & 2.1810358078749115\\
\hline
 53 & 2.161084374990729 & 2.159621468892033 & 2.1831325852783596\\
\hline
 54 & 2.163638906199599 & 2.1621528943817827 & 2.185164204044484\\
\hline
 55 & 2.1661360423984997 & 2.164600746007633 & 2.1871337881796253\\
\hline
 56 & 2.1685455762131505 & 2.1669838446058245 & 2.189044260732053\\
\hline
 57 & 2.17088052166434 & 2.1692927235554205 & 2.1908983599081324\\
\hline
 58 & 2.173156975581998 & 2.17154464226942 & 2.1926986536484128\\
\hline
 59 & 2.1753868074630205 & 2.173733743540678 & 2.194447552833746\\
\hline
 60 & 2.1775580313757423 & 2.1758618646166585 & 2.196147323270252\\
\hline
 61 & 2.1796539024262835 & 2.1779381497638695 & 2.197800096583713\\
\hline
 62 & 2.181684624439103 & 2.1799527763240447 & 2.199407880138139\\
\hline
 63 & 2.183671322526126 & 2.1819134618835685 & 2.200972566079641\\
\hline
 64 & 2.185625664940774 & 2.1838237846732547 & 2.202495939594841\\
\hline
 65 & 2.1875494899093324 & 2.1856892870308204 & 2.203979686462775\\
\hline
 66 & 2.1894023170939043 & 2.1875183857519787 & 2.2054253999702604\\
\hline
 67 & 2.1911918122871388 & 2.18929267351782 & 2.2068345872528328\\
\hline
 68 & 2.192938048928508 & 2.1910205669424956 & 2.208208675116513\\
\hline
 69 & 2.1946541480422255 & 2.1927152484466426 & 2.209549015389644\\
\hline
 70 & 2.196351190908704 & 2.19436251693407 & 2.210856889848733\\
\hline
 71 & 2.1980264943777943 & 2.1959741277496128 & 2.2121335147575847\\
\hline
 72 & 2.199638990992517 & 2.1975551370165314 & 2.2133800450549295\\
\hline
 73 & 2.201195468448455 & 2.1990979377684456 & 2.2145975782220644\\
\hline
 74 & 2.2027167127907497 & 2.2006025251201864 & 2.2157871578588773\\
\hline
 75 & 2.2042098404328647 & 2.2020746269975446 & 2.21694977699375\\
\hline
 76 & 2.2056830630396034 & 2.2035241039015827 & 2.2180863811502745\\
\hline
 77 & 2.2071427242962454 & 2.204929509280723 & 2.219197871191554\\
\hline
 78 & 2.208586806232645 & 2.2063055634755595 & 2.2202851059607602\\
\hline
 79 & 2.2099723195874885 & 2.207659567166601 & 2.221348904734892\\
\hline
 80 & 2.2113220802261164 & 2.2089930212619016 & 2.2223900495070885\\
\hline
 81 & 2.212638959237257 & 2.2102894802454824 & 2.2234092871113598\\
\hline
 82 & 2.213931024081998 & 2.2115554172936966 & 2.2244073312024053\\
\hline
 83 & 2.215201601334469 & 2.2127971269009685 & 2.225384864101958\\
\hline
 84 & 2.216450741223749 & 2.214025083536994 & 2.2263425385221263\\
\hline
 85 & 2.2176831196177855 & 2.215222600886031 & 2.227280979175232\\
\hline
 86 & 2.21890853390639 & 2.216396082132078 & 2.2282007842788385\\
\hline
 87 & 2.2201231708577858 & 2.2175472034019768 & 2.2291025269638842\\
\hline
 88 & 2.2213018305735686 & 2.218679753554799 & 2.229986756593167\\
\hline
 89 & 2.2224395504175263 & 2.219797600470205 & 2.2308539999968273\\
\hline
 90 & 2.2235502739618855 & 2.2208814478953633 & 2.2317047626308755\\
\hline
 91 & 2.2246364545163155 & 2.2219459487326625 & 2.232539529664348\\
\hline
 92 & 2.225703217017416 & 2.222993377806948 & 2.233358767000202\\
\hline
 93 & 2.226755182159646 & 2.2240270358827616 & 2.234162922234636\\
\hline
 94 & 2.22779692784887 & 2.225049621939469 & 2.2349524255591584\\
\hline
 95 & 2.2288278012809983 & 2.2260483368863495 & 2.2357276906093704\\
\hline
 96 & 2.229854847099967 & 2.2270216289339064 & 2.236489115264133\\
\hline
 97 & 2.2308791883475068 & 2.227978242294249 & 2.2372370823984697\\
\hline
 98 & 2.231887015796814 & 2.228922541255891 & 2.237971960593348\\
\hline
 99 & 2.2328644036886414 & 2.2298526426158514 & 2.2386941048051834\\
\hline
 100 & 2.233816394952243 & 2.230774435151576 & 2.2394038569977424\\
\hline
\end{tabular}
\label{tab:full-table-2}
\end{table}

\end{document}